\documentclass[sigplan]{acmart}

\acmConference[]{}{}{}
\acmDOI{}
\acmISBN{}
\acmConference[]{}{}{}

\setcopyright{none}

\usepackage{microtype}
\usepackage{amsmath}
\usepackage{amssymb}
\usepackage{stmaryrd}
\usepackage{graphicx}
\usepackage{paralist}
\usepackage{color}
\usepackage{hyperref}
\usepackage{float}
\usepackage{bchart}
\usepackage{pgfplots}
\usepackage{tikz}
\usetikzlibrary{fit}
\restylefloat{figure}
\usepackage{caption}
\usepackage{subcaption}

\usepackage{algorithm}
\usepackage{algpseudocode}

\usepackage{listings}

\usepackage{listings-golang} 

\makeatletter
\newcommand{\xRightarrow}[2][]{\ext@arrow 0359\Rightarrowfill@{#1}{#2}}
\makeatother

\newcommand{\bi}{\begin{array}[t]{@{}l@{}}}
\newcommand{\ei}{\end{array}}
\newcommand{\ba}{\begin{array}}
\newcommand{\ea}{\end{array}}
\newcommand{\bda}{\[\ba}
\newcommand{\eda}{\ea\]}
\newcommand{\bp}{\begin{quote}\tt\begin{tabbing}}
\newcommand{\ep}{\end{tabbing}\end{quote}}

\newcommand{\ignore}[1]{}

\newcommand{\ms}[1]{{\bf MS: #1}}

\newcommand{\mathem}{\sf}

\newcommand{\lockset}[1]{{\mathem lockset}(#1)}

\newcommand{\evt}[2]{#1_{#2}}
\newcommand{\thread}[2]{#1 \sharp #2}

\newcommand{\SHBEE}{\mbox{SHB$^{\scriptstyle E+E}$}}
\newcommand{\SHBE}{\mbox{SHB$^{\scriptstyle E}$}}  
\newcommand{\SHBP}{\mbox{SHB$^{\forall +}$}}
\newcommand{\SHBFORALL}{\mbox{SHB$^\forall$}}

\newcommand{\SHB}{\mbox{SHB}}

\newcommand{\Tsan}{\mbox{Tsan}}

\newcommand{\e}{\alpha}
\newcommand{\f}{\beta}
\newcommand{\g}{\gamma}

\newcommand{\incC}[2]{{\mathem inc}(#1,#2)}
\newcommand{\maxN}[2]{{\mathem max}(#1,#2)}

\newcommand{\raceVC}[2]{{\mathem raceCheck}(#1,#2)}
\newcommand{\raceWRD}[2]{{\mathem raceWRDCheck}(#1,#2)}

\newcommand{\supVC}[2]{#1 \sqcup #2}
\newcommand{\lteVC}[2]{#1 \sqsubseteq #2}
\newcommand{\accVC}[2]{#1[ #2 ]}
\newcommand{\accTr}[2]{#1[ #2 ]}
\newcommand{\updateVC}[3]{#1[ #2 \mapsto #3 ]}

\newcommand{\pp}{\ \texttt{++}}

\newcommand{\lockE}[1]{\mathit{acq}(#1)} 
\newcommand{\unlockE}[1]{\mathit{rel}(#1)} 
\newcommand{\readE}[1]{r(#1)}
\newcommand{\writeE}[1]{w(#1)}

\newcommand{\readEE}[2]{r(#1)_{#2}}
\newcommand{\writeEE}[2]{w(#1)_{#2}}
\newcommand{\lockEE}[2]{acq(#1)_{#2}}  
\newcommand{\unlockEE}[2]{rel(#1)_{#2}} 
\newcommand{\forkEE}[2]{{\mathit fork}(#1)_{#2}}
\newcommand{\joinEE}[2]{{\mathit join}(#1)_{#2}}

\newcommand{\extract}[2]{#1 \downarrow #2}
\newcommand{\proj}[2]{\mathit{proj}_{\sharp #1}(#2)}
\newcommand{\pos}[1]{\mathit{pos}(#1)}
\newcommand{\posP}[2]{\mathit{pos}_{{\scriptstyle #1}}(#2)}
\newcommand{\compTID}[1]{\mathit{thread}(#1)}
\newcommand{\compTIDP}[2]{\mathit{thread}_{{\scriptstyle #1}}(#2)}

\newcommand{\events}[1]{\mathit{events}(#1)}

\newcommand{\hb}[2]{#1 <^{\scriptscriptstyle HB} #2}
\newcommand{\notHB}[2]{#1 \| #2}

\newcommand{\hbP}[3]{#2 <^{\scriptscriptstyle HB(#1)} #3}

\newcommand{\vcN}{V}                                  
\newcommand{\rwVC}[1]{\mathit{RW}(#1)}
\newcommand{\wVC}[1]{\mathit{W}(#1)}
\newcommand{\writeVC}[1]{\mathit{Write}(#1)}
\newcommand{\lastWriteVC}[1]{\mathit{L_W}(#1)}
\newcommand{\lastWriteID}[1]{\mathit{L_W^{id}}(#1)}

\newcommand{\readVC}[1]{\mathit{Read}(#1)}

\newcommand{\threadVC}[1]{\mathit{Th}(#1)}
\newcommand{\lockVC}[1]{\mathit{Rel}(#1)}

\newcommand{\edges}[1]{\mathit{edges}(#1)}
\newcommand{\oldedges}{\mathit{E}}
\newcommand{\gtEdge}{\prec}

\newcommand{\rwT}[1]{T^{rw}_{#1}}
\newcommand{\rwTx}{\rwT{x}}

\newcommand{\vcEvt}{\mathit{evt}}
\newcommand{\concEvt}[1]{\mathit{conc}(#1)}

\newcommand{\allRaces}{\mathcal{R}^T}
\newcommand{\allConcs}[1]{{\mathcal C}^T(#1)}
\newcommand{\accConcs}[1]{{\mathcal P}(#1)}
\newcommand{\accR}[1]{{\mathcal R}(#1)}

\makeatletter
\newdimen\legendxshift
\newdimen\legendyshift
\newcount\legendlines
\newcommand{\bclldist}{1mm}
\newcommand{\bclegend}[3][10mm]{%
	\legendxshift=0pt\relax
	\legendyshift=0pt\relax
	\xdef\legendnodes{}%
	\foreach \lcolor/\ltext [count=\ll from 1] in {#3}%
	{\global\legendlines\ll\pgftext{\setbox0\hbox{\bcfontstyle\ltext}\ifdim\wd0>\legendxshift\global\legendxshift\wd0\fi}}%
	\@tempdima#1\@tempdima0.5\@tempdima
	\pgftext{\bcfontstyle\global\legendxshift\dimexpr\bcwidth-\legendxshift-\bclldist-\@tempdima-0.72em}
	\legendyshift\dimexpr5mm+#2\relax
	\legendyshift\legendlines\legendyshift
	\global\legendyshift\dimexpr\bcpos-2.5mm+\bclldist+\legendyshift
	\begin{scope}[shift={(\legendxshift,\legendyshift)}]
		\coordinate (lp) at (0,0);
		\foreach \lcolor/\ltext [count=\ll from 1] in {#3}%
		{
			\node[anchor=north, minimum width=#1, minimum height=5mm,fill=\lcolor] (lb\ll) at (lp) {};
			\node[anchor=west] (l\ll) at (lb\ll.east) {\bcfontstyle\ltext};
			\coordinate (lp) at ($(lp)-(0,5mm+#2)$);
			\xdef\legendnodes{\legendnodes (lb\ll)(l\ll)}
		}
		\node[draw, inner sep=\bclldist,fit=\legendnodes] (frame) {};
	\end{scope}
}
\makeatother

\begin{document}

\title{Predicting All Data Race Pairs for a Specific Schedule (extended version)}         




\author{Martin Sulzmann}
\affiliation{
  \department{Faculty of Computer Science and Business Information Systems}
  \institution{Karlsruhe University of Applied Sciences}
  \streetaddress{Moltkestrasse 30}
  \city{Karlsruhe}
  \postcode{76133}
  \country{Germany}
}
\email{martin.sulzmann@hs-karlsruhe.de}

\author{Kai Stadtm{\"u}ller}
\affiliation{
  \department{Faculty of Computer Science and Business Information Systems}
  \institution{Karlsruhe University of Applied Sciences}
  \streetaddress{Moltkestrasse 30}
  \city{Karlsruhe}
  \postcode{76133}
  \country{Germany}
}
\email{kai.stadtmueller@live.de}


\begin{abstract}
  We consider the problem of data race prediction where the program's behavior
  is represented by a trace. A trace is a sequence of program events recorded
  during the execution of the program.
  We employ the schedulable happens-before relation to characterize
  all pairs of events that are in a race for the schedule as manifested in the trace.
  Compared to the classic happens-before relation, the schedulable happens-before relations
  properly takes care of write-read dependencies and thus avoids false positives.
  The challenge is to efficiently identify all (schedulable) data race pairs.
  We present a refined linear time vector clock algorithm
  to predict many of the schedulable data race pairs.
  We introduce a quadratic time post-processing algorithm to predict
  all remaining data race pairs.
  This improves the state of the art in the area and our experiments
  show that our approach scales to real-world examples.
  Thus, the user can systematically examine and fix all program locations that are in a race
  for a particular schedule.
\end{abstract}

\settopmatter{printacmref=false} 

\maketitle

\section{Introduction}

We consider the challenge of data race prediction
for a specific program run.
We assume concurrent programs making use of shared variables
and acquire/release (a.k.a.~lock/unlock) primitives.
We assume that relevant program events such as write/read and
acquire/release operations are recorded in a program trace.
This trace is the basis for our analysis.

The challenge of trace-based data race prediction is that
two conflicting events, e.g.~two writes involving the same variable,
may not necessarily be found next to each other in the trace.
That is, under some interleaving semantics these two writes
happen in sequence but in the actual run-time execution environment
these two writes may happen concurrently to each other.
The actual challenge is to predict a data race by identifying
a valid reordering of the trace under which both conflicting events
appear next to each other in the (reordered) trace.

In this work, we follow the happens-before~\cite{lamport1978time} line of work
for dynamic data race prediction based on a specific program trace.
Our goal is to efficiently identify \emph{all} data race \emph{pairs} for
a \emph{specific schedule}.
This is something that has not been addressed by any prior happens-before based data race predictors.
To illustrate the issue and show the usefulness of our method, we consider the following example.

\begin{verbatim}
// Thread 1              // Thread 2
{                        {
   x = 5;   // E1          acquire(y);  // E3
   x = 6;   // E2          x = 7;       // E4
}                          release(y);  // E5
                         }
\end{verbatim}
We record the trace of events that arises during the interleaved execution of both threads.
We assume a program run where first the writes on shared variable $x$ in thread 1 are executed followed by the
acquire mutex $y$, write $x$ and release mutex $y$ in thread 2.
The resulting trace is of the form $[E_1, E_2, E_3, E_4, E_5]$ where we simply record the program locations
connected to each event.

FastTrack~\cite{flanagan2010fasttrack} and SHB~\cite{Mathur:2018:HFR:3288538:3276515}
are state-of-the art happens-before based race predictors.
For the above trace, both race predictors
only report that $E_4$ is part of a write-write race.
This information is not very helpful as location $E_4$ is protected by a mutex.
Our method reports the pair $(E_2, E_4)$ by enjoying the same $O(n*k)$ time complexity
as the earlier works. Parameter $n$ refers to the size of the trace and parameter $k$ to the number of threads.
Based on the pair $(E_2, E_4)$,
the user can easily see that $E_2$ is unprotected and concludes that this location needs to be fixed.

Assuming some additional processing step, our method is able to identify all pairs of events
that are in a race for the schedule represented by the trace $[E_1, E_2, E_3, E_4, E_5]$.
For the above example, we additionally report the pair $(E_1, E_4)$.
By having such complete diagnostic information, the user can systematically fix the program.
For example, by protecting $E_1$ \emph{and} $E_2$ via a mutex.

The additional processing step requires $O(n*n)$ time.
This extra cost is worthwhile as the more detailed diagnostic information potentially
avoids further incremental re-runs to fix one race after the race.
More seriously, incremental fixing of races might obfuscate some races.

For example, consider a quick incremental fix based on the pair $(E_2, E_4)$
where we guard location $E_2$ via a mutex but $E_1$ remains unguarded.
The resulting program is as follows.
\begin{verbatim}
// Thread 1              // Thread 2
{                        {
   x = 5;     // E1        acquire(y);  // E3
   acquire(y) // E2a       x = 7;       // E4
   x = 6;     // E2b       release(y);  // E5
   release(y) // E2c     }
}                          
\end{verbatim}
Suppose, we re-run the program where we assume a similar interleaved execution as before.
First thread 1 and then thread 2. This leads to
the trace $[E_1, E_{2a}, E_{2b}, E_{2c}, E_3, E_4, E_5]$.
For this trace, both FastTrack and SHB report that there is no race.
The reason for this is that the acquire event at location $E_3$ must happen
after the preceding release event at location $E_{2c}$.
Hence, the write at location $E_1$ appears to happen before the write at location $E_4$.
Hence, there is no data race for this schedule.

However, under a different schedule where thread 2 executes first,
we find that $(E_1, E_4)$ are in a write-write race.
The issue is that the happens-before order relation is sensitive
to the schedule of events as recorded in the trace.
We say that the happens-before order is \emph{trace/schedule-specific}.

Recent works such as~\cite{Kini:2017:DRP:3140587.3062374,Roemer:2018:HUS:3296979.3192385}
attempt to derive \emph{some further} data races
for as many \emph{alternative} schedules as possible.
These works, like FastTrack and SHB,
only report some of the events involved in a race.
Further re-runs to fix races that result from already explored
schedules are necessary.

Our approach is to  report all pairs of events that are in a race for the trace-specific schedule.
This enables the user to systematically examine and fix all races for a specific schedule.
We achieve this via  a novel two-phase data race predictor
where in the first phase we (a) predict as many conflicting pairs of events, and
(b) generate a compact, variable-specific representation of the
happens-before relation to which we refer to as edge constraints.
The second phase uses the reported conflicting pairs of events
and edge constraints to identify
all remaining races via a simple graph traversal.

In summary, our contributions are:
\begin{itemize}
\item We formalize and rigorously verify our 
       two-phase method to identify all races for a trace-specific schedule
      (Sections~\ref{sec:all-concurrent-writes-reads} and~\ref{sec:all-concurrent-writes-reads-post-processing}).

\item We have fully implemented the approach and
  provide for a comparison with the state
   of the art in this area (Section~\ref{sec:implementation}).
\end{itemize}

The upcoming section reviews the idea behind prior happens-before based
data race predictors and highlights the main idea behind our approach.
Background on events, run-time traces
and the happens-before relation is introduced in Section~\ref{sec:events-traces}.
Section~\ref{sec:conclusion} concludes and summarizes related work.

Additional material such as proofs of results stated are given
in the appendix.

\section{Technical Overview}
\label{sec:technical-overview}
\label{sec:our-idea}

\subsection{Data Race Prediction via Vector Clocks}

Vector clocks are a popular method to establish
the happens-before relation among events.
A vector clock $V$ is an array of time stamps (clocks)
where each array position belongs to a specific thread.
For each event we compute its vector clock
and can thus identify the relative order among
events by comparing vector clocks.

\begin{definition}[Vector Clocks]
  A \emph{vector clock} $V$ is a list of \emph{time stamps} of the following form.
  \bda{rcl}
   V  & ::= & [i_1,\dots,i_n]
   \eda
   We assume vector clocks are of a fixed size $n$.
   Time stamps are natural numbers and each time stamp position $j$ corresponds to the thread
   with identifier $j$.

We define vector clock $V_1$ to be smaller than vector clock $V_2$,
written $V_1 < V_2$,
if (1) for each thread $i$, $i$'s time stamp in $V_1$ is smaller or equal
compared to $i$'s time stamp in $V_2$, and
(2) there exists a thread $i$ where $i$'s time stamp in $V_1$
is strictly smaller compared to $i$'s time stamp in $V_2$.
\end{definition}
If the vector clock assigned to event $e$ is smaller
compared to the vector clock assigned to $f$,
then we can argue that $e$ happens before $f$.

\begin{figure}[tp]
\bda{lll|ll|l}
  & \thread{1}{} & [1,0] & \thread{2}{} & [0,1] &   \writeVC{x} 
\\ \hline
1. & \writeE{x} & [1,0] && & [1,0] 
\\
2. & \writeE{x} & [2,0] && & [2,0] 
\\
3. &            & & \writeE{x} & [0,1] & [2,1] 
\eda
  \caption{Vector Clock Construction (\SHB)}
    \label{fig:vc-construction-simple}
\end{figure}

Figure~\ref{fig:vc-construction-simple}
shows an example of how vector clocks are constructed. 
We consider a trace with events from two threads.
The trace resembles the trace from the introduction
where for brevity we omit acquire/release events.

Events are recorded in linear order as they take place during
program execution.
For each event, we record the position in the trace as well as
the id of the thread in which the event took place.
We use a tabular notation to record this information.

For the first thread we find a write followed by another write.
In the trace, the write in the second thread appears after
the writes in the first thread.
As the writes lack any synchronization,
we conclude that
$(\writeEE{x}{1}, \writeEE{x}{3})$
and $(\writeEE{x}{2}, \writeEE{x}{3})$
are two pairs of events that represent a race.

Let us carry out the vector clock construction steps.
The initial vector clock for thread 1 is $[1,0]$
and for thread 2 it is $[0,1]$.
For each event, we record the vector clock when processing the event.
After processing, we increment the time stamp of the thread.
We ignore the column $\writeVC{x}$ for the moment.

How to check for a write-write data race?
Assuming that each event carries a vector clock, we simply
need to compare the vector clocks of events.
For example, $[1,0] \not < [0,1]$ and $[1,0] \not > [0,1]$.
Hence, there is a trace reordering where
we can place $\writeEE{x}{1}$ and $\writeEE{x}{3}$
right next to each other.

The above reasoning implies that
(1) for each event we need to store its vector clock, and
(2) consider all possible combinations of events that might form
a race and then compare their vector clocks.
This requires time $O(n*n*k)$
where $n$ is the size of the trace and $k$ the number of threads.
A rather costly computation and therefore
data race predictors such as FastTrack
and \SHB\ perform a different approach
that only requires time $O(n*k)$.

\SHB\ keeps track of all writes that took place
via the vector clock $\writeVC{x}$.
For the first write $\writeEE{x}{1}$,
we simply set $\writeVC{x}$ to $[1,0]$.
Each subsequent write synchronizes with $\writeVC{x}$.
Before synchronization, each write checks
if its vector clock is greater or equal
to the vector clock recorded in $\writeVC{x}$.
If not, we report a race, as the current write
must be concurrent to an earlier write.

Hence, SHB reports that $\writeEE{x}{3}$ is part of a data race.
But SHB neither reports the full pair that represents the race
such as $(\writeEE{x}{2}, \writeEE{x}{4})$,
nor reports all events that are part of a race.
The same observations applies to FastTrack.

\begin{figure}[tp]
\bda{lll|ll|ll}
  & \thread{1}{} & [1,0] & \thread{2}{} & [0,1] &   \wVC{x} & \oldedges
\\ \hline
1. & \writeE{x} & [1,0] && & \{ \thread{1}{1} \} & 
\\
2. & \writeE{x} & [2,0] && & \{ \thread{1}{2} \} &  \{ \thread{1}{1} \gtEdge \thread{1}{2} \}
\\
3. &            & & \writeE{x} & [0,1] & \{ \thread{1}{2}, \thread{2}{1} \} &
\eda
  \caption{Epochs and Edge Constraints (\SHBEE)}
    \label{fig:epochs-edges-simple}
\end{figure}

\subsection{ Our Idea: Epochs and Edge Constraints}

We make two major adjustments to the \SHB\ algorithm
to compute (1) pairs of events that are in a race and (2) all such pairs for the given schedule.
We refer to the resulting algorithm as \SHBEE. We explain the adjustments based on our running example.

First, instead of using a vector clock $\writeVC{x}$ to maintain
the concurrent writes,
we use a set $\wVC{x}$ of time stamps per concurrent write (referred to as epoch).
Thus, we can report \emph{some} conflicting \emph{pairs of events} that form a data race
while still guaranteeing the $O(n*k)$ time complexity like \SHB.
Some data race pairs are still missing.

Our second adjustment of the SHB algorithm builds up a set of edge constraints.
Each time a write happens after some of the currently
recorded writes in $\wVC{x}$ and edge is added among the two writes involved.
By traversing edges starting from an existing race pair,
we report \emph{all} conflicting \emph{pairs of events} that form a data race.
This (post-processing) step takes time $O(n*n)$ as we will show in detail later.

We first consider the issue of reporting a pair
of events that form a race, instead of just reporting
some event that is part of race as it is done in \SHB.
To achieve this more refined reporting,
we use epochs.

An epoch is a pair of thread id and time stamp for that thread.
\begin{definition}[Epoch]
  Let $j$ be a thread id and $k$ be a time stamp.
  Then, we write $\thread{j}{k}$ to denote an \emph{epoch}.
\end{definition}
Each event can be uniquely associated to an epoch.
Take its vector clock and extract the time stamp $k$ for the thread $j$
the event belongs to. For each event this pair of information
represents a unique key to locate the event.

We revisit our earlier example.
See Figure~\ref{fig:epochs-edges-simple}.
We ignore the component $\oldedges$ for the moment.
In the first step, we add $\writeEE{x}{1}$'s epoch $\thread{1}{1}$
to $\wVC{x}$.
We maintain the invariant that $\wVC{x}$
is the set of recently processed writes that are concurrent to each other.
Hence, after processing the second write we find $\wVC{x}$ to be equal
$\thread{1}{2}$ as $\thread{1}{1} < \thread{1}{2}$.

Consider the processing of $\writeEE{x}{3}$ where $\thread{2}{1}$
is its epoch and $[0,1]$ is its vector clock.
The time stamp of thread 1 for vector clock $[0,1]$
is smaller than the time stamp recorded by epoch $\thread{1}{2}$.
Hence, we argue that the write represented by $\thread{1}{2}$
is concurrent to the write represented by $\thread{2}{1}$.
Then, $\wVC{x}$ becomes equal to $\{ \thread{1}{2}, \thread{2}{1} \}$.

We conclude that the associated events are in a write-write race
and report the write-write race pair $(\writeEE{x}{2}, \writeEE{x}{3})$.
So, via epochs we can provide more refined data race reports.
The complexity remains the same compared to the original
\SHB\ algorithm. The set of epochs grows as much as $O(k)$.
Hence, we require time $O(n * k)$
to report the same conflicting events as \SHB\ \emph{but}
we additionally also report the complete conflicting pair of events.

To compute all races, we require one further adaptation of the \SHB\
algorithm. Each time we replace an epoch in $\wVC{x}$
by an epoch that happens later, we add an edge
from the epoch to its replacement.

\begin{definition}[Edge]
  Let $\thread{i}{k}$ and $\thread{j}{k}$ be two epochs.
  Then, we write $\thread{i}{k} \gtEdge \thread{j}{k}$
  to denote the \emph{edge} from
  $\thread{i}{k}$ to $\thread{j}{k}$.
  We sometimes refer to $\thread{i}{k} \gtEdge \thread{j}{k}$
  as an \emph{edge constraint}.
  We write $\gtEdge^*$ to denote the transitive closure
  among edge constraints.
\end{definition}

These edge constraints are collected in $\oldedges$.
In our example, we add
$\thread{1}{1} \gtEdge \thread{2}{2}$.
Edge constraints represent a condensed view
of the happens-before relation restricted to a specific variable.

We employ edge constraints in a post-processing phase
to identify all missing race pairs.
When checking for further races starting
with an existing race pair $(e,f)$,
we look for events $g$ that are reachable via edge constraints
from~$e$ and~$f$.

If $g \gtEdge^* e$ and $g \gtEdge^* f$,
then neither $(g,e)$ nor $(g,f)$ form another pair of conflicting events.
This is the case because events in edge constraint relation
are in the happens-before relation.
If $g \gtEdge^* e$ and $g \not \gtEdge^* f$,
then $(g,f)$ is a potential pair of conflicting events.
Potential because while edge constraints are sound
but they are not complete w.r.t.~the happens-before relation
as we will explain in detail later.

For our running example, we have the race pair
$(\writeEE{x}{2}, \writeEE{x}{3})$ and find
$\thread{1}{1} \gtEdge^* \thread{2}{1}$.
Hence, we conclude that $(\writeEE{x}{1}, \writeEE{x}{3})$
is another race pair.
We can show that we thus compute
all race pairs for the trace-specific schedule
and the construction takes time $O(n*n)$.

Experiments show that our approach is effective
and provides much more detailed diagnostic information
to the user compared to \SHB\ and FastTrack.

\section{Events and Run-Time Traces}
\label{sec:events-traces}

We introduce some background material on events and run-time traces.
We consider concurrent programs that
make use of threads, lock-based primitives, acquire and release of a mutex,
and shared memory reads and writes.
We assume that reads and writes follow the sequential consistency memory model.

Execution of a program yields a trace.
A trace is a sequence of events that took place
and represents the interleaved execution of the various threads found in the program.
Below, we formalize the shape of a trace and the kind of events we consider.
Details of how to obtain a trace are discussed
in the later Section~\ref{sec:implementation}.

\begin{definition}[Run-Time Traces and Events]
\label{def:run-time-traces-events}  
\bda{lcll}
  T & ::= & [] \mid \thread{i}{e} : T   & \mbox{Trace}
  \\ e,f,g & ::= &  \readEE{x}{j}
           \mid \writeEE{x}{j}
           \mid \lockEE{y}{j}
           \mid \unlockEE{y}{j}

           & \mbox{Events}
\eda
\end{definition}

A trace $T$ is a list of events. We adopt Haskell notation for lists
and assume that the list of objects $[o_1,\dots,o_n]$ is a shorthand
for $o_1:\dots:o_n:[]$. We write $\pp$ to denote the concatenation operator among lists.
For each event $e$, we record the thread id number $i$ in which the event took place,
written $\thread{i}{e}$, and the position $j$ of the event in the trace.
We sometimes omit the thread id and position for brevity.

We write $\readEE{x}{j}$ and $\writeEE{x}{j}$
to denote a read and write event on shared variable $x$ at position $j$.
We write $\lockEE{y}{j}$ and $\unlockEE{y}{j}$
to denote a lock and unlock event on mutex $y$ at position $j$.
For brevity, we omit intra-thread synchronization primitives such as fork and join.
They are dealt with by our implementation and do not pose any challenges
in terms of the underlying theory as their
treatment is very similar to acquire and release.
For details see Appendix~\ref{sec:fork-join}.

For trace $T$, we assume some helper functions to access the thread id and position of $e$. 
We define $\compTIDP{T}{e} = j$ if $T=T_1 \pp\ [\thread{j}{e}] \pp\ T_2$ for some traces $T_1, T_2$.
We define $\posP{T}{\readEE{x}{j}} = j$ and so on.
We assume that the trace position is \emph{accurate}:
If $\posP{T}{e} = n$ then $T=\thread{i_1}{e_1}: \dots : \thread{i_{n-1}}{e_{n-1}} : \thread{i}{e} : T'$
for some events $\thread{i_k}{e_k}$ and trace $T'$.
We sometimes drop the component $T$ and
write $\compTID{e}$ and $\pos{e}$ for short.

Given a trace $T$, we can also access an event at a certain position $k$.
We define $\accTr{T}{k} = e$ if $e \in T$ where $\posP{T}{e} = k$.

For trace $T$, we define $\events{T} = \{ e \mid \exists T_1,T_2,j. T = T_1 \pp [\thread{j}{e}] \pp T_2 \}$
to be the set of events in $T$.
We write $e \in T$ if $e \in \events{T}$.

For trace $T$, we define $\proj{i}{T} = T'$ the projection of $T$ onto thread $j$
where (1) for each $e \in T$ where $\compTIDP{T}{e} = i$ we have that $e \in T'$, and
(2) for each $e, f \in T'$ where $\posP{T'}{e} < \posP{T'}{f}$ we have that $\posP{T}{e} < \posP{T}{f}$.
That is, the projection onto a thread comprised of all events in that thread
and the program order remains the same.

Besides accurate trace positions, we demand that
acquire and release events are in a proper acquire/release order.
For each acquire there must be a matching release in the same thread
and atomic sections cannot overlap. In case an acquire event lacks
a matching release because the program has been terminated prematurely,
we assume a dummy release event.

\begin{definition}[Proper Acquire/Release Order]
\label{def:proper-acq-rel-order}  
  We say a trace $T$ enjoys a \emph{proper acquire/release order}
  iff the following conditions are satisfied:
  \begin{itemize}
   \item For each $\thread{i}{\lockEE{y}{j_1}} \in T$
  there exists $\thread{i}{\unlockEE{y}{j_2}} \in T$
  where $j_1 < j_2$. For the event with the smallest position $j_2$,
  we have that no other acquire/release event on $y$
  occurs in between trace positions $j_1$ and $j_2$.
   \item For each $\thread{i}{\unlockEE{y}{j_2}} \in T$
  there exists $\thread{i}{\lockEE{y}{j_1}} \in T$
  where $j_1 < j_2$.
  For the event with the greatest position $j_1$,
  we have that no other acquire/release event on $y$
  occurs in between trace positions $j_1$ and $j_2$.
  \end{itemize}
\end{definition}

We say a trace $T$ is \emph{well-formed} iff
trace positions in $T$ are accurate and $T$
enjoys a proper acquire/release order.

Each well-formed trace implies a happens-before relation
among events. We follow~\cite{Mathur:2018:HFR:3288538:3276515}
and employ a happens-before relation that
guarantees that all trace reorderings that satisfy
this happens-before relation are schedulable.

\begin{definition}[Schedulable Happens-Before~\cite{Mathur:2018:HFR:3288538:3276515}]
\label{def:happens-before}  
  Let $T$ be a well-formed trace.
  We define a relation $\hbP{T}{}{}$ among trace events
  as the smallest partial order such that the following holds:
  \begin{description}
  \item[Program order (PO):]
    Let $e, f \in T$. Then, $\hbP{T}{e}{f}$ iff
    $\compTID{e} = \compTID{f}$ and $\pos{e} < \pos{f}$.
  \item[Write-read dependency (WRD):]
    Let $\readEE{x}{j}, \writeEE{x}{k} \in T$.
    Then, $\hbP{T}{\writeEE{x}{j}}{\readEE{x}{k}}$ iff
    $j < k$ 
    and
    for all $e\in T$ where
    $j < \pos{e}$ 
    and $\pos{e} < k$ 
    we find that $e$ is not a write event on $x$.
  \item[Release-acquire dependency (RAD):] \mbox{}
    \\
    Let $\unlockEE{y}{j}, \lockEE{y}{k} \in T$. Then, $\hbP{T}{\unlockEE{y}{j}}{\lockEE{y}{k}}$ iff
    $j < k$  where
    $\compTID{\unlockEE{y}{j}} \not= \compTID{\lockEE{y}{k}}$ and
    for all $e \in T$ where $j < \pos{e}$, $\pos{e} < k$
    and $\compTID{\unlockEE{y}{j}} \not= \compTID{e}$ we find that $e$
    is not an acquire event on $y$.
  \end{description}
  We refer to $\hbP{T}{}{}$ as the \emph{schedulable happens-before}
  relation. We generally say happens-before for short.
  If trace $T$ is fixed based on the context,
  we write $\hb{}{}$ for short.

  We say two events $e, f \in T$ are \emph{concurrent} to each other
  if neither $\hbP{T}{e}{f}$, nor $\hbP{T}{f}{e}$ holds.
  In such a situation, we write $\notHB{e}{f}$.
\end{definition}

Like earlier happens-before relations,
e.g.~consider~\cite{flanagan2010fasttrack},
we demand that events must be ordered
based on the program order (PO).
An acquire event
must happen after the nearest release event in the trace (RAD).
As observed in~\cite{Mathur:2018:HFR:3288538:3276515},
write-read dependencies (WRD) must be respected.
Otherwise, events are not necessarily schedulable.
Consider the following example
taken from~\cite{Mathur:2018:HFR:3288538:3276515}.

\begin{example}
  Consider the well-formed trace
  $$T = [\thread{1}{\readEE{x}{1}}, \thread{1}{\writeEE{y}{2}},
    \thread{2}{\readEE{y}{3}}, \thread{2}{\writeEE{x}{4}}]$$
  Based on the schedulable happens-before relation, we have that
  $\hbP{T}{\writeEE{y}{2}}{\readEE{y}{3}}$.
  In combination with the program order,
  we find find that $\hbP{T}{\readEE{x}{1}}{\writeEE{x}{4}}$.
  So, the read on $x$ happens before the write on $x$.
  Hence, there is no race that involves variable $x$.

  This is a conservative approximation of the program's behavior.
  The specific value $y$ read at trace position $3$ may
  affect the program's control flow. Hence, the earlier
  write on $y$ must remain in the (relative) same position
  with respect to the subsequent read.

    The FastTrack algorithm~\cite{flanagan2010fasttrack}
  ignores WRD relations and therefore may yield false positives.
  Without WRD the following trace reordering is possible:
  $
 [ \thread{2}{\readEE{y}{3}}, \thread{2}{\writeEE{x}{4}},
     \thread{1}{\readEE{x}{1}}, \thread{1}{\writeEE{y}{2}} ]
 $.
 But there is no schedule resulting from some program run
  under which we encounter this sequence of events.
\end{example}

To summarize.
The schedulable happens-before relation ensures that all data races
are indeed schedulable. We give a more precise description
of data races.

\begin{definition}[Read/Write Events]
  Let $T$ be a well-formed trace.
  We define $\rwTx$ as the set of all read/write
  events in $T$ on some variable $x$.  

  Let $M \subseteq T$ be a subset of events in $T$.
  Then, we define $\extract{M}{\rwTx} = M \cap \rwTx$.
\end{definition}

\begin{definition}[Data Races]
  \label{def:data-race}
  Let $T$ be a well-formed trace.
  Let $x$ be some variable and $e, f \in \rwTx$ be
  two read/write events on $x$.
  
  We say that $(e, f)$ are in a \emph{write-write data race}
  if $e$ and $f$ are both write events
      and $e$ and $f$ are concurrent to each other.

   We say that $(e, f)$ are in a \emph{write-read data race}
   if $e$ is a write event and $f$ is a read event
   where either (1) $e$ and $f$ are concurrent to each other, or
   (2a) $\hbP{T}{e}{f}$ and
   (2b) $\neg\exists e' \in T. \hbP{T}{e}{e'} \wedge \hbP{T}{e'}{f}$, and
   (2c) $\compTID{e} \not = \compTID{f}$.

   We denote by $\allRaces$ the set of all
   pairs of events $(e,f)$ where $e, f \in T$
   and $(e,f)$ are in write-write or write-read data race relation.
\end{definition}  
Our definition of a data race implies that the trace can be reordered
such that both (conflicting) events appear next to each other in the trace.
This clearly applies for races where events involved are concurrent to each other.
In case of a write-read dependency, the read must follow the write.
This leads to a race if no other event appears in between (2b),
and the write and read take place in different threads (2c).

\begin{example}
  Consider the following trace where we use a tabular notation.
  Events belonging to a specific threads appear in a separate column.
  \bda{ll|l|l}
  & \thread{1}{} & \thread{2}{} & \thread{3}{}
  \\ \hline
  1. & \writeE{x} &&
  \\
  2. & & \writeE{x} &
  \\
  3. & & \readE{x} &
  \\
  4. & & & \readE{x}   
  \eda
  We find that
  $\allRaces = \{
  (\writeEE{x}{1}, \writeEE{x}{2}),
  (\writeEE{x}{1}, \readEE{x}{3}),
  (\writeEE{x}{2}, \readEE{x}{4}) \}$
  where we omit symmetric cases.
  For the pairs $(\writeEE{x}{1}, \writeEE{x}{2})$ and
  $(\writeEE{x}{1}, \readEE{x}{3})$, the events involved are concurrent to each other.
  The pair $(\writeEE{x}{2}, \readEE{x}{4})$ is in write-read dependency relation
  and satisfies conditions (2a-c)
  Definition~\ref{def:data-race}.
  The pair $(\writeEE{x}{2}, \readEE{x}{3})$ is also in a write-read dependency
  relation but only satisfies (2b) and not (2c).
  Hence, this pair is not part of $\allRaces$.
\end{example}  

\section{SHB Algorithm}
\label{sec:standard-vc}

We give a recast of the \SHB\ vector clock algorithm
introduced in~\cite{Mathur:2018:HFR:3288538:3276515}
to predict data races under the schedulable happen-before relation.
In addition, we show how to extend \SHB\ to identify races
due to write-read dependencies.
   
Algorithm~\ref{alg:vc-race-checker}, referred to as the \SHB\ algorithm,
processes events in a stream-based fashion.
The algorithm maintains several vector clocks.

For each thread $i$ we maintain a vector clock $\threadVC{i}$.
For each variable $x$, we use a vector clock $\writeVC{x}$
to maintain the write access history to $x$.
Similarly, we use $\readVC{x}$ to maintain the read access history.
We use a vector clock $\lastWriteVC{x}$ to maintain the last write access
as in the order specified in the trace.
Similarly, for each mutex $y$, we use
vector clock $\lockVC{y}$ to maintain the last
release event on $y$.

Initially, for each vector clock $\threadVC{i}$
all time stamps are set to 0 but position $i$ where the time stamp
is set to 1.
For $\writeVC{x}$, $\readVC{x}$, $\lastWriteVC{x}$ and $\lockVC{y}$
all time stamps are set to 0.

We define some helper functions
to access and update the time stamp of a specific thread as well as
 a (point-wise) join  operation of two vector clocks.
   \bda{lcl}
   \accVC{[i_1,\dots,i_{j-1}, i_j, i_{j+1}, \dots, i_n]}{j} & = & i_j
   \\ \updateVC{[i_1,\dots,i_n]}{j}{k} & = & [i_1,\dots,i_{j-1}, k, i_{j+1}, \dots, i_n]
   \\   \supVC{[i_1,\dots, i_n]}{[j_1,\dots,j_n]} & = & [\maxN{i_1}{j_1},\dots,\maxN{i_n}{j_n}]
   \eda
It is easy to see that the join operation is associative.
Hence, we will write $\supVC{V_1}{\supVC{V_2}{V_3}}$
as a short-hand for $\supVC{V_1}{(\supVC{V_2}{V_3})}$.   

We write $\incC{V}{i}$ as a short-hand for $V := \updateVC{V}{i}{\accVC{V}{i}+1}$.
We write
$$\raceVC{V_1}{V_2}$$
as a short-hand for
``if $\neg \lteVC{V_1}{V_2}$ then race detected''
where $\lteVC{[i_1,\dots, i_n]}{[j_1,\dots,j_n]}$ iff $i_1 \leq j_1 \wedge \dots \wedge i_n \leq j_n$.

We consider the various cases of Algorithm~\ref{alg:vc-race-checker}.
For acquire and release events, parameter $i$ refers to the thread id
and $x$ refers to the name of the mutex.
Similarly, for writes and reads, $i$ refers to the thread id
and $x$ refers to the name of the variable.
Parameter $k$ refers to the trace position.
This is parameter is only necessary for reads and writes.
We use the trace position to uniquely identify each event.

Vector clocks are updated as follows.
In case of an acquire event
we synchronize the thread's vector clock with the most recent (prior) release
event by building the union of the vector clocks
$\threadVC{i}$ and $\lockVC{x}$.
In case of a release event, we update $\lockVC{x}$.

In case of a write event, 
we compare the thread's vector clock
against the read and write histories
$\readVC{x}$ and $\writeVC{x}$ to check for
a write-read and write-write race.
Then, we update $\lastWriteVC{x}$ to record the vector clock
of the most recent write on $x$.
For the write history $\writeVC{x}$, we update the time stamp
at position $i$ to the thread's time stamp at that position.

In case of a read event,  check for read-write races by comparing
the thread's vector clock against $\writeVC{x}$.
Only then we synchronize the thread's vector clock
with the vector clock $\lastWriteVC{x}$ of the most recent write.
The history of reads is updated similarly as in case of writes.

\begin{algorithm}
  \caption{\SHB\ algorithm}\label{alg:vc-race-checker}

\begin{algorithmic}[1]
\Procedure{acquire}{$i,x$}
\State $\threadVC{i} = \supVC{\threadVC{i}}{\lockVC{x}}$
\EndProcedure
\end{algorithmic}

\begin{algorithmic}[1]
\Procedure{write}{$i,x, k$}
\State $\raceVC{\writeVC{x}}{\threadVC{i}}$
\State $\raceVC{\readVC{x}}{\threadVC{i}}$  
\State $\lastWriteVC{x} = \threadVC{i}$
\State $\updateVC{\writeVC{x}}{i}{\accVC{\threadVC{i}}{i}}$
\State $\incC{\threadVC{i}}{i}$
\EndProcedure
\end{algorithmic}

\begin{algorithmic}[1]
\Procedure{release}{$i,x$}
\State $\lockVC{x} = \threadVC{i}$  
\State $\incC{\threadVC{i}}{i}$
\EndProcedure
\end{algorithmic}

\begin{algorithmic}[1]
\Procedure{read}{$i,x, k$}
\State $\raceVC{\writeVC{x}}{\threadVC{i}}$    
\State $\threadVC{i} = \supVC{\threadVC{i}}{\lastWriteVC{x}}$
\State $\updateVC{\readVC{x}}{i}{\accVC{\threadVC{i}}{i}}$
\State $\incC{\threadVC{i}}{i}$
\EndProcedure
\end{algorithmic}

\end{algorithm}

\begin{example}
\label{ex:shb-incomplete}  
  We consider a run of the \SHB\ algorithm.
  The example does not involve
  any mutex and the last write always takes place in the same thread.
  Hence, the components $\lockVC{x}$
  and $\lastWriteVC{x}$ can be ignored.
  We underline events for which a call to $\raceVC{}{}$ issues a race.
  The subscript indicates if the event is in a race with
  a read (r) or a write (w).

  For presentation purposes, we first show the annotated trace
  where the columns $\writeVC{x}$ and $\readVC{x}$
  follow below.
  
    \bda{lll|ll|ll}
  & \thread{1}{} & [1,0,0]
  & \thread{2}{} & [0,1,0]
  & \thread{3}{} & [0,0,1]
  \\ \hline
  1. & \writeE{x} & [1,0,0]
     &&
     &&
  \\
  2. & \readE{x} & [2,0,0]
  &&
  &&
  \\
  3. &&
  & \underline{\writeE{x}}_{rw} & [0,1,0] &
  &
  \\
  4. &&
  & \underline{\readE{x}}_{w} & [0,2,0] &
  \\
  5. &&
  &&
  & \underline{\readE{x}}_{w} & [0,0,1]
  \eda

  \bda{lll}
  & \writeVC{x}
  & \readVC{x}
  \\ \hline
  1. 
    & [1,0,0] &
  \\
  2. 
    && [2,0,0] 
  \\
  3. 
    & [1,1,0] &
  \\
  4. 
   && [2,2,0]
  \\
  5. 
  && [2,2,1]
  \eda

  We first find a write followed by a read
  and update the thread's vector clock as well as
  $\writeVC{x}$ and $\readVC{x}$ accordingly.
  As both events are in the same thread, there is no race issued.

  In the third step, we find another write.
  The event is underlined with $rw$ as
  both calls
  $$\raceVC{\writeVC{x}}{[0,1,0]}$$
  and
  $$\raceVC{\readVC{x}}{[0,1,0]}$$
  issue a race.

  The read in  the fourth step is in a race with a write.
  The same applies to the read in the fifth step.

  We observe that events $\writeEE{x}{1}$
  and $\readEE{x}{2}$ are part of a race but not underlined.
  For example, the vector clock $[1,0,0]$ of $\writeEE{x}{1}$
  and the vector clock $[0,1,0]$ of $\writeEE{x}{3}$
  are incomparable. Hence, both events form a write-write data race pair.
\end{example}

Based on the above example,
we conclude that the \SHB\ algorithm reports
\emph{some} events that are involved in a race but not \emph{all}.

The time and space complexity of the \SHB\ algorithm 
is $O(n*k)$ where $n$ is the length of the trace and $k$ the number of threads.
We assume that each vector clock requires $O(k)$ space
and comparing two vector clocks takes time $O(k)$.
For each event $n$ we maintain a constant number of vector clocks including
some comparisons.  Hence, $O(n*k)$.

\section{\SHBEE\ Algorithm}
\label{sec:all-concurrent-writes-reads}

Our goal is to identify all pairs of events that are in a race
and identify all such conflicting pairs for the given schedule.
We focus on conflicting pairs where the events involved are concurrent
to each other. The special case
of write-read races due to write-read dependencies is dealt with by
the adaptation described above.
To achieve our goal we require two phases.

\begin{algorithm}
  \caption{\SHBEE\ Algorithm}\label{alg:vc-edges-conc-read-write}

\begin{algorithmic}[1]
\Procedure{acquire}{$i,x$}
\State $\threadVC{i} = \supVC{\threadVC{i}}{\lockVC{x}}$
\EndProcedure
\end{algorithmic}

\begin{algorithmic}[1]
\Procedure{release}{$i,x$}
\State $\lockVC{x} = \threadVC{i}$  

\State $\incC{\threadVC{i}}{i}$
\EndProcedure
\end{algorithmic}

\begin{algorithmic}[1]  
\Procedure{write}{$i,x, k$}
\State $\vcEvt = \{ (k,\threadVC{i}) \} \cup \vcEvt$  
\State $\concEvt{x} = 
    \{ (\thread{j}{k}, \thread{i}{\accVC{\threadVC{i}}{i}})
    \mid \thread{j}{k} \in \rwVC{x}
          \wedge k > \accVC{\threadVC{i}}{j} \} \cup \concEvt{x}$
\State $\edges{x} = \edges{x} \cup
       \{ \thread{j}{k} \gtEdge \thread{i}{\accVC{\threadVC{i}}{i}}
       \mid \thread{j}{k} \in \rwVC{x} \wedge
       k < \accVC{\threadVC{i}}{j} \}$
\State $\rwVC{x} = \{ \thread{i}{\accVC{\threadVC{i}}{i}} \}       
         \cup \{ \thread{j}{k} \mid \thread{j}{k} \in \rwVC{x} \wedge
         k > \accVC{\threadVC{i}}{j} \}$
\State $\lastWriteVC{x} = \threadVC{i}$         
\State $\incC{\threadVC{i}}{i}$
\EndProcedure
\end{algorithmic}

\begin{algorithmic}[1]
\Procedure{read}{$i,x, k$}
\State $\threadVC{i} = \supVC{\threadVC{i}}{\lastWriteVC{x}}$
\State $\vcEvt = \{ (k,\threadVC{i}) \} \cup \vcEvt$  
\State $\concEvt{x} = 
    \{ (\thread{j}{k}, \thread{i}{\accVC{\threadVC{i}}{i}})
    \mid \thread{j}{k} \in \rwVC{x} 
          \wedge k > \accVC{\threadVC{i}}{j} \} \cup \concEvt{x}$  
\State $\edges{x} = \edges{x} \cup
       \{ \thread{j}{k} \gtEdge \thread{i}{\accVC{\threadVC{i}}{i}}
       \mid \thread{j}{k} \in \rwVC{x} \wedge
            k < \accVC{\threadVC{i}}{j} \}$
\State $\rwVC{x} = \{ \thread{i}{\accVC{\threadVC{i}}{i}} \}
       \cup \{ \thread{j}{k} \mid \thread{j}{k} \in \rwVC{x} \wedge
            k > \accVC{\threadVC{i}}{j} \}$
\State $\incC{\threadVC{i}}{i}$
\EndProcedure
\end{algorithmic}

\end{algorithm}

The first phase is carried out by Algorithm~\ref{alg:vc-edges-conc-read-write},
referred to as the \SHBEE\ algorithm.
The (second) post-processing phase is described in the upcoming section.

Like \SHB, algorithm \SHBEE\ employs vector clocks $\threadVC{i}$, $\lockVC{x}$
and $\lastWriteVC{x}$. In addition, \SHBEE\ outputs three types of sets
where two sets are indexed by shared variable $x$:
$\concEvt{x}$, $\edges{x}$ and $\vcEvt$.
Set $\concEvt{x}$ holds pairs of concurrent reads/writes.
As motivated in Section~\ref{sec:our-idea}, via a linear pass through the events
the set $\concEvt{x}$ might not necessarily include all pairs of events that are in a race.
However, the missing pairs are reachable via edge constraints accumulated by $\edges{x}$.
Edge constraints are sound but not complete w.r.t.~the happens-before relation.
Hence, we might need to filter out some candidate pairs.
For filtering, we require the set $\vcEvt$ where $\vcEvt$ records for each event its vector clock,
The details of post-processing based on these three sets are explained in the upcoming section.

Besides the three sets $\concEvt{x}$, $\edges{x}$ and $\vcEvt$, \SHBEE\ maintains some other set $\rwVC{x}$
to record the most recent concurrent set of reads/writes.
Initially, all sets are empty.

The treatment of acquire and release is the same as in case of \SHB.
In case of a write event, we record the writer's vector clock by updating
$\vcEvt$. We use the trace position to uniquely identify each event
and thus record its associated vector clock
as pairs in $\vcEvt$.
We record the most recent write by updating $\lastWriteVC{x}$.
Each read event, synchronizes with $\lastWriteVC{x}$ to ensure
that write-read dependencies are respected.

Next, we consider the update of sets $\concEvt{x}$ and $\rwVC{x}$.
Recall that each event can be identified by its epoch and vice versa.
When processing event $e$ in thread $i$ with vector clock $\threadVC{i}$,
the epoch associated to $e$ is $\thread{i}{\accVC{\threadVC{i}}{i}}$.
We add new pairs to $\concEvt{x}$
by comparing the current epoch (event) $\thread{i}{\accVC{\threadVC{i}}{i}}$
against epochs $\thread{j}{k}$ in $\rwVC{x}$.
A new pair is added if $k > \accVC{\threadVC{i}}{j}$.
Similarly, we adjust the set $\rwVC{x}$.
The epoch $\thread{i}{\accVC{\threadVC{i}}{i}}$
is added to $\rwVC{x}$ and we only keep epoch $\thread{j}{k}$
in the set $\rwVC{x}$ 
if $k > \accVC{\threadVC{i}}{j}$.
The treatment is the same for writes and reads.

For each read and write, we add an edge from
an epoch $\thread{j}{k}$ in $\rwVC{x}$
to the currently being processed
epoch $\thread{i}{\accVC{\threadVC{i}}{i}}$
if $k < \accVC{\threadVC{i}}{j}$.

\subsection{Properties}

The following result establishes that 
the events (epochs) reported in $\rwVC{x}$
and the pairs of events (epochs) reported in $\concEvt{x}$ are concurrent
to each other.

\begin{proposition}
\label{prop:time-stamp-vc-concurrent}  
  Let $T$ be a well-formed trace and $x$ some variable.
  Then, for any subtrace $T' \subseteq T$
  the sets $\concEvt{x}$ and $\rwVC{x}$ obtained
  after running Algorithm~\ref{alg:vc-edges-conc-read-write} on $T'$
  enjoy the following properties.
  All events in $\rwVC{x}$ are concurrent to each other
  and for each $(e,f) \in \concEvt{x}$ we have that $e$ and $f$ are
  concurrent to each other.
\end{proposition}

We note that by construction,
for each $(e,f) \in \concEvt{x}$
we have that $e$ appears before $f$ in the trace.
That is, $\pos{e} < \pos{f}$.

As shown by the following example,
the set $\concEvt{x}$ does not necessarily contain  all concurrent pairs of events for
the schedule implied by $\hbP{T}{}{}$.

\begin{example}
\label{ex:conc-missed}  
  Consider the following annotated trace after
  processing via Algorithm~\ref{alg:vc-edges-conc-read-write}.
  For clarity, epochs are annotated with their corresponding events,
  e.g.~$\evt{\thread{1}{1}}{w_1}$.

  \bda{ll|l|lll}
    & \thread{1}{} & \thread{2}{} & \rwVC{x} & \concEvt{x}
  \\ \hline 1. & \writeE{x} [1,0] &        & \{ \evt{\thread{1}{1}}{w_1} \} &
  \\ 2. & \writeE{x} [2,0] &        & \{ \evt{\thread{1}{2}}{w_2} \} &
  \\ 3. &                  & \writeE{x} [0,1] & \{ \evt{\thread{1}{2}}{w_2}, \evt{\thread{2}{1}}{w_3} \} &
     \{ (\evt{\thread{1}{2}}{w_2}, \evt{\thread{2}{1}}{w_3}) \}
  \eda  

  Algorithm~\ref{alg:vc-edges-conc-read-write} reports the concurrent
  pair $(\evt{\thread{1}{2}}{w_2}, \evt{\thread{2}{1}}{w_3})$
  but fails to report the (missing) concurrent pair
  $(\evt{\thread{1}{1}}{w_1}, \evt{\thread{2}{1}}{w_3})$.
\end{example}  

While $\concEvt{x}$ lacks certain pairs of concurrent events,
we can provide for a sufficient condition
under which a pair $(e,f)$ is added to $\concEvt{x}$.
In essence, a pair $(e,f)$
is added to $\concEvt{x}$ whenever $f$ has no other concurrent partner
that appears within $e$ and $f$ in the trace.

\begin{lemma}
\label{le:no-other-conc-in-between}
Let $T$ be a well-formed trace.
Let $e,f \in \rwTx$ for some variable $x$
such that (1) $e$ and $f$ are concurrent to each other,
(2) $\pos{f} > \pos{e}$, and
(3) $\neg\exists g \in \rwTx$ where $g$ and $f$ are concurrent
to each other and $\pos{f} > \pos{g} > \pos{e}$.
Let $\concEvt{x}$ be the set obtained by running
Algorithm~\ref{alg:vc-edges-conc-read-write} on $T$
Then, we find that $(e,f) \in \concEvt{x}$.
\end{lemma}

Recall Example~\ref{ex:conc-missed}.
The reported pair $(w_2, w_3)$ satisfies this property.
However, $(w_1, w_3)$ is not reported because
$w_2$ appears in between.
As we show in the up-coming
section, such missing pairs can be reached via edge constraints
because edges approximate the happens-before relation.

\begin{proposition}[Soundness of Edge Constraints]
\label{le:edge-sound}  
  Let $T$ be a well-formed trace and $x$ be some variable.
  Let $\edges{x}$ be the set of edge constraints obtained by
  running Algorithm~\ref{alg:vc-edges-conc-read-write} on $T$.
  Then, for each $e,f \in T$ where $e \gtEdge^* f$
  based on the edges in $\edges{x}$ we find that $\hbP{T}{e}{f}$.
\end{proposition}

\subsection{Time and Space Complexity}

Let $n$ be the size of the trace $T$
and $k$ be the number of threads.
We assume that the number of distinct variables $x$ is a constant.
We consider the time and space complexity of running \SHBEE.

The size of the vector clocks and the set $\rwVC{x}$ is
bounded by $O(k)$.
In each step of \SHBEE,
adjustment of vector clocks takes time $O(k)$.
Adjustment of sets $\concEvt{x}$, $\edges{x}$ and $\rwVC{x}$
requires to consider $O(k)$ epochs where each
comparison among epochs is constant.
So, in each step this requires time $O(k)$.
Adjustment of set $\vcEvt$ takes constant time.
Overall, Algorithm~\ref{alg:vc-edges-conc-read-write}
runs in time $O(n*k)$.

We consider the space complexity.
The sets $\vcEvt$, $\concEvt{x}$ and $\edges{x}$
take space $O(n*k)$.
This applies to $\vcEvt$ because
for each event the size of the vector clock is $O(k)$.
Each element in $\concEvt{x}$ and $\edges{x}$ requires
constant space. In each step, we may add $O(k)$ new elements
because the size of $\rwVC{x}$ is bounded by $O(k)$.
Overall, Algorithm~\ref{alg:vc-edges-conc-read-write}
requires space $O(n*k)$.

\section{\SHBEE\ Post-Processing}
\label{sec:all-concurrent-writes-reads-post-processing}

Based on the sets $\concEvt{x}$ and $\edges{x}$ computed by
\SHBEE\ we compute all remaining concurrent reads/writes.
The important property is that all pairs of concurrent reads/writes
are either already contained in $\concEvt{x}$ or can be reached via some edges.

\begin{definition}[All Concurrent Reads/Writes]
\label{def:conc-reads-writes}  
   Let $T$ be a well-formed trace and $x$ be some variable.
   We define $\allConcs{x} = \{ (e, f) \mid
   e, f \in \rwTx \wedge \notHB{e}{f} \wedge \pos{e} < \pos{f} \}$
   the set of all reads/writes on $x$ that are concurrent
   to each other.
\end{definition}  
It is clear that if $(e,f)$ is a concurrent pair,
so is the pair $(f,e)$. For technical reasons,
we only keep the pair where the first component
appears first in the trace.

\begin{lemma}
\label{le:reach-all-pairs}
  Let $T$ be a well-formed trace and $x$ be some variable.
  Let $\concEvt{x}$ be the set of concurrent pairs of events
  and $\edges{x}$ be the set of edge constraints obtained by
  running Algorithm~\ref{alg:vc-edges-conc-read-write} on $T$.
  Let $(e,f) \in \allConcs{x}$ where $(e,f) \not\in \concEvt{x}$
  and $\pos{e} < \pos{f}$.
  Then, there exists $g$ such that $e \gtEdge^* g$
  and $(g,f) \in \concEvt{x}$.
\end{lemma}

Based on results stated in Lemmas~\ref{le:no-other-conc-in-between}
and~\ref{le:reach-all-pairs}, we can effectively
compute all concurrent writes/reads by scanning
through the sets $\concEvt{x}$ and $\rwVC{x}$.

We use symbols $\e, \f, \g$ to denote epochs.
As we know, each epoch uniquely corresponds to an event and vice versa.
Hence, for epoch $\e$, we write $\pos{\e}$ to obtain
the trace position of the event that corresponds to $\e$.

\begin{definition}[\SHBEE\ Post-Processing]
\label{def:all-concs-post}  
  Let $\concEvt{x}$ and $\edges{x}$ be obtained
  by running Algorithm~\ref{alg:vc-edges-conc-read-write}
  on some well-formed trace $T$.

  We first introduce a total order among pairs of epochs
  $(\e,\f) \in \concEvt{x}$ and $(\e',\f') \in \concEvt{x}$.
  We define $(\e,\f) < (\e',\f')$ if $\pos{\e} < \pos{\e'}$.
  This defines a total order among all pairs of epochs
  where the event corresponding to the epoch in the first
  position appears before the event corresponding
  to the epoch in the second position.
  
  Then, repeatedly perform the following steps where we initially
  assume that $\accConcs{x} := \{ \}$.
  
  \begin{enumerate}
  \item If $\concEvt{x} = \{\}$ stop.
  \item Otherwise, let $(\e,\f)$ be the smallest element in $\concEvt{x}$.
  \item Let $G=\{\g_1,\dots,\g_n\}$ be maximal such that
    $\g_1 \gtEdge \e, \dots, \g_n \gtEdge \e \in \edges{x}$ and
    $\pos{\g_1} < \dots < \pos{\g_n}$.
  \item $\accConcs{x} := \{ (\e,\f) \} \cup \accConcs{x}$.
  \item $\concEvt{x} := \{ (\g_1,\f), \dots, (\g_n,\f) \} \cup (\concEvt{x} - \{ (\e,\f) \})$.
  \item Repeat.
\end{enumerate}    

\end{definition}  

\begin{theorem}
  \label{th:all-concs-post}
  Let $T$ be a well-formed trace of size $n$.
  Let $x$ be a variable.
  Then, construction of $\accConcs{x}$ takes time $O(n*n)$
  and $\allConcs{x} \subseteq \accConcs{x}$.
\end{theorem}

The set $\accConcs{x}$ is a superset of $\allConcs{x}$
because pairs $(\e, \f)$ added to $\accConcs{x}$
may not necessarily be concurrent to each other.
For space reasons, we refer
to Appendix~\ref{sec:edges-overapproximate} for an example.
We can easily eliminate such cases
by comparing $\e$'s time stamp against the time stamp of $\f$'s vector clock.

\begin{proposition}[Eliminate Non-Concurrent Pairs]
\label{prop:elim-non-conc-pairs}  
  Let $x$ be some variable.
  Let $(\thread{i}{k}, \thread{j}{l}) \in \accConcs{x}$
  Let $\pos{\thread{j}{l}} = m$ and $(m, V) \in \vcEvt$.
  Then, remove $(\thread{i}{k}, \thread{j}{l})$ from $\accConcs{x}$
  if $k < \accVC{V}{j}$.
  Applying this check to all pairs in $\accConcs{x}$
  yields that $\accConcs{x} = \allConcs{x}$.
\end{proposition}

\begin{example}
  \label{ex:shb-recall}
  We consider a run of \SHBEE.
  We omit vector clocks and the component $\vcEvt$.
  Instead of epochs, we refer to the corresponding event
  and its trace position.
  For $\concEvt{x}$ we write down the
  newly added elements in each step. For space reasons,
  we omit an extra column for $\edges{x}$
  and report this set below.
  \bda{ll|l|l|ll}
  & \thread{1}{}
  & \thread{2}{}
  & \thread{3}{}
  & \rwVC{x}
  & \concEvt{x}
  \\ \hline
  1. & \writeE{x}
  &
  &
  & \{ \writeEE{x}{1} \}
  &
  \\
  2. & \readE{x}
  &
  &
  & \{ \readEE{x}{2} \}
  &
  \\
  3. &
  & \writeE{x}
  &
  & \{ \readEE{x}{2}, \writeEE{x}{3} \}
  &  ( \readEE{x}{2}, \writeEE{x}{3}) 
  \\
  4. &
  & \readE{x}
  &
  & \{ \readEE{x}{2}, \readEE{x}{4} \}
  & (\readEE{x}{2}, \readEE{x}{4})  
  \\
  5. &
  &
  & \readE{x}
  & \{ \readEE{x}{2}, \readEE{x}{4}, \readEE{x}{5} \}
  & (\readEE{x}{2}, \readEE{x}{5})
  \\
   &&&&
    (\readEE{x}{4}, \readEE{x}{5})
  \eda
  where
  $\edges{x} =
  \{ \writeEE{x}{1} \gtEdge \readEE{x}{2}
  , \writeEE{x}{3} \gtEdge \readEE{x}{4} \}$.
  
  So, after processing the trace
  we obtain
  \bda{lcl}
  \concEvt{x} & = & \{
           (\readEE{x}{2}, \writeEE{x}{3}),
           (\readEE{x}{2}, \readEE{x}{4}),
  \\ & &
           (\readEE{x}{2}, \readEE{x}{5}),
           (\readEE{x}{4}, \readEE{x}{5})
  \}
  \eda
 and
 $$\edges{x} = \{
 \writeEE{x}{1} \gtEdge \readEE{x}{2},
 \writeEE{x}{3} \gtEdge \readEE{x}{4}
 \}$$

 We apply the post-processing described in Definition~\ref{def:all-concs-post} 
 to obtain the set $\accConcs{x}$.
 Elimination as described in Proposition~\ref{prop:elim-non-conc-pairs}
 is not necessary for this example.

 Via $(\readEE{x}{2}, \writeEE{x}{3})$ and edge $\writeEE{x}{1} \gtEdge \readEE{x}{2}$
 we add the pair $(\writeEE{x}{1}, \writeEE{x}{3})$.
 Via similar reasoning, we finally obtain
 \bda{lcl}
  \accConcs{x} & = & \{ (\readEE{x}{2}, \writeEE{x}{3}),
           (\readEE{x}{2}, \readEE{x}{4}),
           (\readEE{x}{2}, \readEE{x}{5}),
  \\ && (\readEE{x}{4}, \readEE{x}{5}),
  (\writeEE{x}{1}, \writeEE{x}{3}),
        (\writeEE{x}{1}, \readEE{x}{4}),
   \\ &&       (\writeEE{x}{1}, \readEE{x}{5}),
        (\writeEE{x}{3}, \readEE{x}{5})
  \}
 \eda
\end{example}  

\subsection{Data Race Optimization}
\label{sec:all-data-races}

The set $\concEvt{x}$ and the construction in Definition~\ref{def:all-concs-post}
also includes pairs of concurrent reads. Such pairs are generally not interesting
as they do not represent a data race.
However, we cannot simply ignore read-read pairs.
Recall Example~\ref{ex:shb-recall}
where via the pair $(\readEE{x}{2}, \readEE{x}{5})$
we obtain $(\writeEE{x}{1}, \readEE{x}{5})$.

\begin{algorithm}
\caption{\SHBEE\ Algorithm (Data Race Optimization)}\label{alg:vc-edges-conc-read-write-opt}

\begin{algorithmic}[1]
\Procedure{read}{$i,x, k$}
\State $\threadVC{i} = \supVC{\threadVC{i}}{\lastWriteVC{x}}$
\State $\vcEvt = \{ (k,\threadVC{i}) \} \cup \vcEvt$  
\State $\concEvt{x} = 
    \{ (\thread{j}{k}, \thread{i}{\accVC{\threadVC{i}}{i}})
    \mid \thread{j}{k} \in \rwVC{x} 
          \wedge k > \accVC{\threadVC{i}}{j} \wedge \mbox{$\thread{j}{k}$ is a write} \} \cup \concEvt{x}$  
\State $\edges{x} = \edges{x} \cup
       \{ \thread{j}{k} \gtEdge \thread{i}{\accVC{\threadVC{i}}{i}}
       \mid \thread{j}{k} \in \rwVC{x} \wedge
            k < \accVC{\threadVC{i}}{j} \}$
\State $\rwVC{x} = \{ \thread{i}{\accVC{\threadVC{i}}{i}} \}
       \cup \{ \thread{j}{k} \mid \thread{j}{k} \in \rwVC{x} \wedge
            (k > \accVC{\threadVC{i}}{j} \vee \mbox{$\thread{j}{k}$ is a write}  \}$
\State $\incC{\threadVC{i}}{i}$
\EndProcedure
\end{algorithmic}

\end{algorithm}

To ignore read-read pairs during the post-processing construction of all (concurrent) data races,
we need to adapt our method as follows.
In \SHBEE, in case of a read event~$e$, we only remove events from $\rwVC{x}$
if these events happen before~$e$ \emph{and} are not write events.
That is, a write event in $\rwVC{x}$ can only be removed by a subsequent write.
See Algorithm~\ref{alg:vc-edges-conc-read-write-opt}
where we only show the case of read as all other parts are unaffected.

Thus, we effectively build the transitive closure of writes that can be reached via a read.
Hence, any concurrent write-read pair $(e,f)$ that would be obtained
via edge $f \gtEdge g$ and the read-read pair $(f,g)$,
is already present in $\concEvt{x}$.
Hence, there is no need to record read-read pairs in $\concEvt{x}$.

\begin{example}
  We consider the trace from Example~\ref{ex:shb-recall}
  where we use the adapted Algorithm~\ref{alg:vc-edges-conc-read-write-opt}.
  We omit the trace and show only components $\rwVC{x}$ and $\concEvt{x}$.
   \bda{l|ll}
   & \rwVC{x}
   & \concEvt{x}
   \\ \hline
   1. 
   & \{ \writeEE{x}{1} \}
   &
   \\
   2. 
   & \{ \writeEE{x}{1}, \readEE{x}{2} \}
   &
   \\
   3. 
   & \{ \writeEE{x}{1}, \readEE{x}{2}, \writeEE{x}{3} \}
   & (\writeEE{x}{1}, \writeEE{x}{3})
   \\
   &&  ( \readEE{x}{2}, \writeEE{x}{3})
   \\
   4. 
   & \{ \writeEE{x}{1}, \readEE{x}{2}, \writeEE{x}{3}, \readEE{x}{4} \}
   & (\writeEE{x}{1}, \readEE{x}{4})
   \\
   5. 
   & \{ \writeEE{x}{1}, \readEE{x}{2}, \writeEE{x}{3}, \readEE{x}{4}, \readEE{x}{5} \}
   & (\writeEE{x}{1}, \readEE{x}{5})
   \\
   && (\writeEE{x}{3}, \readEE{x}{5})
   \eda
  

  Unlike in the earlier calculations in Example~\ref{ex:shb-recall},
  $\writeEE{x}{1}$ is not eliminated from $\rwVC{x}$ in the second step.
  Hence, we add $(\writeEE{x}{1}, \writeEE{x}{3})$ to $\concEvt{x}$.
  Similarly, in the fourth step we keep $\writeEE{x}{3}$
  and therefore add $(\writeEE{x}{3}, \readEE{x}{5})$ to $\concEvt{x}$.

 After processing the trace we obtain
 \bda{lcl}
   \concEvt{x} & = & \{
 (\writeEE{x}{1}, \writeEE{x}{3}),
 ( \readEE{x}{2}, \writeEE{x}{3}),
   (\writeEE{x}{1}, \readEE{x}{4}),
 \\ & & 
 (\writeEE{x}{1}, \readEE{x}{5}),
 (\writeEE{x}{3}, \readEE{x}{5})
   \}
 \eda

 Post-processing does not yield any further data races for this example.
\end{example}

\subsection{Time and Space Complexity}

We investigate the complexity to predict all data races
that are schedulable as defined by Definition~\ref{def:data-race}.

Algorithm~\ref{alg:vc-edges-conc-read-write-opt}, optimized for data race prediction,
enjoys the same time and space complexity results as \SHBEE.
This is the case because the size of $\rwVC{x}$ is still bounded by $O(k)$.
A read no longer removes a write, so for $k$ concurrent reads we may
have a maximum of $k$ additional writes in $\rwVC{x}$.
Hence, $O(k)$ elements in the worst-case.
Hence, the time complexity of the first phase is $O(n*k)$.

It is easy to integrate the prediction
of all write-read races due to write-read dependencies
without affecting the time and space complexity of Algorithms~\ref{alg:vc-edges-conc-read-write}
and~~\ref{alg:vc-edges-conc-read-write-opt}.

We consider the (second) post-processing phase
where we only consider the time complexity.
Based on the data race optimization describe above, read-read pairs no longer arise.
In practice, this is a huge improvement.
The construction of $\accConcs{x}$ still takes time $O(n*n)$
as stated by Theorem~\ref{th:all-concs-post}
because the worst-case number of pairs to consider remains the same.
The elimination step as described
in Proposition~\ref{prop:elim-non-conc-pairs}
requires to compare time stamps. This takes constant time.
Then, the post-processing phase takes time $O(n*n)$.

In summary, the prediction of all schedulable data races
based on \SHBEE\ and post-processing takes time $O(n*k) + O(n*n)$.

\section{Implementation and Experiments}
\label{sec:implementation}

We have implemented \SHBEE\
including post-processing to obtain all data race pairs.
For comparison, we have also implemented 
\SHB~\cite{Mathur:2018:HFR:3288538:3276515}.
We have also implemented a variant of \SHB\ referred to as \SHBFORALL.
\SHBFORALL\ additionally builds the set $\vcEvt$ that records for each event its vector clock.
In some post-processing phase, we build the set of all data races pairs
by comparing vector clocks for each candidate pair.
This variant has been briefly sketched in~\cite{Mathur:2018:HFR:3288538:3276515}.
The post-processing requires time $O(n*n*k)$ as there are $O(n*n)$ candidate pairs to consider
and each comparison takes time $O(k)$.

The input for all algorithms is a trace in CSV format containing
read, write, acquire, release as well as fork and join events.
Instead of events, we report the code locations connected to each event.
We avoid reporting of repeated code locations.
The same applies to pairs of code locations.

All algorithms are implemented in the Go programming language.
It is in terms of syntax and performance similar to C but offers
garbage collection, memory safety and CSP-style concurrency.
The implementations can be found at \url{https://github.com/KaiSta/shbee}.

\subsection{Benchmarks}

\begin{table*}
    \small
    \begin{tabular}{c|c|c|c|c|c|c|c}
        \hline
        \textbf  & moldyn  & raytracer & xalan & lusearch & tomcat & avrora & h2  \\ \hline
        \textbf{TReplay} (s) & 71 & 0 & 85 & 3 & 37 & 19 & 131\\
        Memory (mb) & 9514 & 46 & 11052 & 576 & 5200 & 2672 & 28168 \\\hline
        
        \textbf{\SHBEE} (s) & 117 & 1 & 105 & 4 & 63 & 53 & 262 \\        
        Memory (mb) & 19415 & 81 & 12566 & 579 & 10404 & 12669 & 59674\\
        Phase1+2 (s) & 64+2 & 0+0 & 11+0 & 1+0 & 24+0 & 36+0 & 135+0 \\
        \#Races & 18+6 &  1+0 & 44+5 & 24+0 & 677+324 & 32+0 & 285+1\\ 
        \hline
        
%
        \textbf{\SHBFORALL} (s) & >1h* & 18 & 1206 & 4 & 957 & >1h* & >1h*\\
        Memory (mb) & 16594 & 69 & 12691 & 585 & 8986 & 4787 & 42420\\
        \#Races & 0+21 &  0+1  & 0+49  & 0+24  & 0+1001  & 0+30 & 0+254\\
        \hline
        
        \textbf{\SHB} (s) & 81 & 0 & 92 & 5 & 54 & 28 & 202\\
        Memory (mb) & 11492 & 55 & 11114 & 586 & 6614 & 3585 & 30072\\ 
        \#Races* & 9 & 1 & 35 & 23 & 492 & 20 &  105 \\\hline

        
        
    \end{tabular}
    \caption{Benchmark Results}
    \label{tab:bench:races}
\end{table*}


We benchmark the performance of \SHBEE, \SHB\ and \SHBFORALL.
For benchmarking we use two Intel Xeon E5-2650 and 64 gb of RAM with Ubuntu 18.04 as operating system.
Following~\cite{Mathur:2018:HFR:3288538:3276515},
we use tests from the Java Grande Forum (\cite{smith2001parallel}) and from the DaCapo (version 9.12, \cite{Blackburn:2006:DBJ:1167473.1167488}) benchmark suite.
All benchmark programs are written in Java and use up to 58 threads.
For instrumentation and tracing we make use of the RoadRunner tool~\cite{flanagan2010roadrunner}.
The entire trace is kept in memory and then (off-line) each algorithm processes the trace.

Table \ref{tab:bench:races} shows the benchmark results.
We measure the time, memory consumption and number of predicted data races.
The first row for each algorithm contains the overall execution time. This includes the start up of the program,
parsing the trace etc. We use the standard `time' program in Ubuntu to measure this time.
The memory consumption is also measured for the complete program and not only for the single algorithms.
We include `TReplay' to measure the time (seconds) and memory consumption (megabytes)
for trace replay without any attached race prediction algorithm.
For example, for the xalan benchmark TReplay takes time 85s and 11052mb.
For \SHB\ we measure 92s and 11114mb. The difference is the actual time and space spent
for data race prediction applying Algorithm~\ref{alg:vc-race-checker}.

For \SHBEE{} we additionally provide the time for each (data race prediction) phase. 
Phase~1 corresponds to running Algorithm~\ref{alg:vc-edges-conc-read-write}
and phase~2 to the post-processing step described in Section~\ref{sec:all-concurrent-writes-reads-post-processing}.
Additional processing steps such as parsing, reporting data races and so on are not included.

For example, the xalan benchmark takes 85 seconds with TReplay while \SHBEE{} takes 105 seconds overall and 11 seconds for phase~1.
The time 0s in phase~2 arises because for most benchmarks data races arise very early during the execution of benchmark programs.
Hence, the post-processing step only needs to cover small portions of the trace and the time spent is negligible.
There is a difference of nine seconds, if we compare the 85 seconds with TReplay against \SHBEE{}'s 105 seconds minus 11 seconds for phase 1+2.
The difference arises because TReplay does not report any data races nor keeps track of any other statistical data.

For the moldyn benchmark, \SHBEE{} requires 64 seconds whereas \SHB's race prediction phase requires 10 seconds
(by subtracting TReplay's 71 seconds from the 81 seconds overall \SHB\ running time).
For the other benchmarks, \SHBEE{}'s phase~1 is comparable with \SHB.
The difference for the moldyn benchmark arises because we predict a huge number of race pairs where the events
involved all refer to the same code locations. For moldyn, \SHBEE{} only reports 18 unique race pairs in phase~1.
The management and elimination of duplicates causes some overhead in \SHBEE. This is something we plan to optimize in the future.

In terms of time and memory consumption, \SHB\ performs best followed.
\SHBEE\ shows competitive performance compared to \SHB.
\SHBFORALL{} does not seem to scale for larger benchmarks.
For moldyn, avrora and h2, we aborted the test after one hour (marked with *).

We consider the number of predicted races.
\SHBEE\ and \SHBFORALL\ report complete pairs of code locations that are in a race.
For both, the number of predicted data races is written as X+Y
where X are the data races found in the first phase and Y those found in the post-processing phase.
In case of \SHBFORALL, X always equals 0 as the calculation of races is completely carried by the post-processing phase.
As can be seen, for most benchmarks, a large portion of race pairs
are already predicted in the first phase via \SHBEE.

If we are only interested in the race pairs predicted in phase~1, we could optimize \SHBEE\ as follows.
We drop $\vcEvt$ and $\edges{x}$ and only maintain $\rwVC{x}$.
Then, we achieve $O(k)$ memory consumption
and it is possible to run \SHBEE\ online (like \SHB).

\subsection{\SHB\ versus \SHBEE}

We carry out a more detailed analysis between \SHB\ and \SHBEE\
regarding the quality and quantity of races reported.
We examine the following questions.
How often does \SHB\ report locations $f$ that are already protected by some mutex?
Recall the discussion from the introduction. Reporting a protected location $f$
is not very helpful in fixing the data race. Rather, we wish to identify the race partner $e$
that lacks protection.

Another interesting question is the following.
How many additional race locations are detected by \SHBEE\ and its post-processing phase?
As argued in the introduction, knowing all race locations is useful in systematically
fixing a buggy schedule.

We consider the first question.
Based on Theorem 4.2. from~\cite{Mathur:2018:HFR:3288538:3276515}, we
observe that for each race location $f$ reported by \SHB\ we have
that \SHBEE\ reports the race pair $(e,f)$ for some location $e$.
Hence, it suffices to consider race pairs $(e,f)$ reported by \SHBEE\ in the following refined analysis.

We employ a variant of \SHBEE\ where for each
read/write event $e$  we compute the set of locks (mutexes) that have been acquired
by the thread by the time we process~$e$.
We refer to this set as $\lockset{e}$. Computation of $\lockset{e}$ is straightforward.
Each thread $i$ maintains $\lockset{i}$.
Each acquire in thread $i$ adds the respective mutex to $\lockset{i}$.
Each release removes the mutex from $\lockset{i}$.
When processing write/read event $e$ in thread $i$, we set $\lockset{e} = \lockset{i}$.

For each data race pair $(e, f)$ reported by \SHBEE\
we must have that $\lockset{e} \cap \lockset{f} = \emptyset$.
To check if $e$, or $f$ or both $e$ and $f$ lack protection, we distinguish
among the following cases:
(C1) $\lockset{e} = \{\}$ and $\lockset{f} \not= \{\}$,
(C2) $\lockset{e} \not= \{\}$ and $\lockset{f} = \{\}$, and
(C3) $\lockset{e} \not= \{\}$, $\lockset{f} \not= \{\}$ and $\lockset{e} \cap \lockset{f} = \{ \}$.

As \SHB\ only reports $f$ for a pair $(e,f)$, case (C1) means that the location reported
by \SHB\ is rather useless for fixing the data race.
In case of (C2) the location reported by \SHB\ is sufficient to fix the race (assuming the fix
involves a mutex). For case (C3), we might also need to inspect the race partner~$e$.

\begin{table}
    \begin{tabular}{l|l|l|l|l}
        Test & C1 & C2 & C3 & \#Race pairs $(e,f)$ \\ \hline
        moldyn & 0 & 0 & 18 & 18\\
        tomcat & 43 & 28 & 606 & 677\\
        xalan  & 1 & 0 & 43 & 44\\
        raytracer & 1 & 0 & 0 & 1\\
        lusearch & 6 & 0 & 18 & 24\\
        avrora & 13 & 4 & 15 & 32\\
        h2 & 8 & 7 & 270 & 285
    \end{tabular}
    \caption{Lockset analysis}
    \label{tab:lockset:results}
\end{table}

We carry out this additional analysis for our benchmark programs
and measure how often the various cases (C1-3) arise.
Results are reported in Table~\ref{tab:lockset:results}.
As can be seen, case (C3) arises most frequently.
Case (C1) arises in general less frequent.
Thanks to the refined analysis provided by our method,
the user can more easily navigate to the source location
that requires fixing.

We examine the second question.
How many additional race locations are detected by \SHBEE\ and its post-processing phase in comparison to \SHB?
For example, consider the trace $[\thread{1}{\writeEE{x}{1}}, \thread{1}{\writeEE{x}{2}}, \thread{2}{\writeEE{x}{3}}]$
for the example in Section~\ref{sec:technical-overview}.
Assuming that trace positions refer to code locations,
\SHB\ reports the location~3. \SHBEE\ reports the locations~2 and~3 because the race pair (2,3) is detected.
The post-processing phase of \SHBEE\ additionally reports location~1 as
post-processing yields the race pair~(1,3).

\begin{table}
    \begin{tabular}{l|c|c|c}
        Test & \SHB\ & \SHBEE\ & \SHBEE\ + post-processing \\ \hline
        moldyn & 9 & 23 & 24\\
        tomcat & 492 & 644 & 709\\
        xalan & 35 & 55 & 55\\
        raytracer & 1 & 1 & 1\\
        lusearch & 23 & 39 & 39\\
        avrora & 20 & 28 & 28\\
        h2 &  105 & 123 & 123
    \end{tabular}
    \caption{Number of distinct locations reported}
    \label{tab:postprocessing:results}
\end{table}

Table~\ref{tab:postprocessing:results} shows these additional analysis results for our benchmark programs.
As can been seen, \SHBEE\ alone (first phase) covers more code locations than \SHB\
and these locations are mostly all locations that are involved in a data race.

\section{Related Works and Conclusion}
\label{sec:related-works}
\label{sec:conclusion}

Vector clocks are the main technical method to establish the happens-before ordering.
Originally, vector clocks were introduced
in the message-passing setting, see works
by Fidge~\cite{fidge1988timestamps} and Mattern~\cite{Mattern89virtualtime}.
In the context of data race prediction,
vector clocks are employed
by Pozniansky and Schuster~\cite{Pozniansky:2003:EOD:966049.781529}.
The FastTrack algorithm by Flanagan and Freund~\cite{flanagan2010fasttrack} employs
an optimized representation of vector clocks where
only the thread's time stamp, referred to as an epoch, need to be traced.
It is folklore knowledge, that vector clock based race predictors
are only sound for the first data race found.
This is due to improper treatment of write-read dependencies that leads
to an overapproximation of the happens-before relation.

ThreadSanitizer (\Tsan)
by Serebryany and Iskhodzhanov~\cite{serebryany2009threadsanitizer}
is a hybrid race predictor that
combines happens-before (for fork-join)
with lockset~\cite{Dinning:1991:DAA:127695:122767}
to identify conflicting memory accesses.
Like FastTrack, \Tsan\ yields potentially false positives.

Banerjee, Bliss, Ma and Petersen~\cite{banerjee2006theory} develop
criteria which data races can be found with a limited (trace) history.
The history options considered  are keeping track of the last event (independent if its a read or a write),
last read and write event, last event for each thread and last write and all concurrent reads.
They observe that none of the limited histories is able to predict all data races.

Mathur, Kini and Viswanathan~\cite{Mathur:2018:HFR:3288538:3276515}
introduce a variant of the standard vector clock algorithm
that properly deals with write-read dependencies.
The thus strengthened happens-before relation is referred to
as the schedulable happens-before relation.
The algorithm in~\cite{Mathur:2018:HFR:3288538:3276515} identifies
some events that are involved in a data race.
The question of how to efficiently infer (all) pairs of events that are in a schedulable data race
is not addressed.

There are several recent works that employ
happens-before methods to derive \emph{further} data races
for as many \emph{alternative} schedules as possible.
See the works Smaragdakis,Evans, Sadowski, Yi and Flanagan~\cite{Smaragdakis:2012:SPR:2103621.2103702}, Kini, Mathur and Viswanathan~\cite{Kini:2017:DRP:3140587.3062374}
as well as Roemer, Gen\c{c} and Bond~\cite{Roemer:2018:HUS:3296979.3192385}.

Huang, Luo and Rosu~\cite{Huang:2015:GGP:2818754.2818856} go a step
further to obtain even more races and trace
values to guarantee that write-read dependencies are respected.
They employ SMT-based solving methods to enumerate as many races as possible.
Recent work by Kalhauge and Palsberg~\cite{Kalhauge:2018:SDP:3288538.3276516}
follows a similar approach. The issue with these methods is that the computational cost
is very high. See the benchmark results reported in~\cite{Kalhauge:2018:SDP:3288538.3276516}.


We attack a different problem that is complementary to the above works.
For a \emph{trace-specific} schedule, we wish to efficiently find \emph{all} pairs of events
that are in a race and that are schedulable
w.r.t.~the happens-before relation defined in~\cite{Mathur:2018:HFR:3288538:3276515}.
Thus, the user is able to systematically examine and fix all data races
for a specific schedule.
Our experiments show that the approach is effective
and provides the user with detailed diagnostic information.




\section*{Acknowledgments}

We thank referees for OOPSLA'19 and MPLR'19
for their helpful comments on previous versions of this paper.

\bibliography{main}


\appendix

\section{Proofs}
\label{sec:proofs}

\subsection{Auxiliary Results}

Based on our construction of the happens-before relation derived
from the trace, we can state that a later in the trace appearing
read/write event can never happen before an earlier in the trace appearing
read/write event on the same variable.

\begin{lemma}
\label{le:happens-before-vs-trace-position}  
  Let $T$ be a well-formed trace.
  Let $e,f \in \rwTx$ such that $\pos{e} > \pos{f}$.
  Then, $\neg (\hbP{T}{e}{f})$.
\end{lemma}
The statement follows by construction. See Definition~\ref{def:happens-before}.

\begin{lemma}[Criteria for Write-Read Dependency Races]
\label{le:criteria-wrd-race}  
  Let $T$ be a well-formed trace.
  Let $(e,f) \in \allRaces$ where $e$ is a write, $e$ a read on
  some variable $x$
  and $(e,f)$ is a write-read data race
  satisfying the criteria (2a-c) in Definition~\ref{def:data-race}.
  Let $j = \compTID{e}$.
  Then, when processing event $f$ we find that
  $\neg (\accVC{\lastWriteVC{x}}{j} \leq \accVC{\threadVC{i}}{j})$.
\end{lemma}  
\begin{proof}
  Event $e$ is processed before $f$.
  Consider the vector clock $\lastWriteVC{x}$ of $e$
  and $\threadVC{i}$ of $f$ (before synchronization with $\lastWriteVC{x}$).
  Suppose $\accVC{\lastWriteVC{x}}{j} \leq \accVC{\threadVC{i}}{j}$.
  This implies that there must have been some form of synchronization
  via $e$'s and $f$'s thread.
  This contradicts the assumptions
  (2b) and (2c) from Definition~\ref{def:data-race}.
  Hence, we have
  that $\neg (\accVC{\lastWriteVC{x}}{j} \leq \accVC{\threadVC{i}}{j})$.  
\end{proof}

\subsection{Proof of Proposition~\ref{prop:shb-completeness-wrd-race}}

\begin{proof}
  Suppose $e$ is the write and $f$ the read event
  where the variable involved is named~$x$.
  Event $e$ is processed before $f$.
  Suppose $j = \compTID{e}$.
  The write history $\writeVC{x}$ is updated by adjusting the time stamp
  at position $j$.
  Based on Lemma~\ref{le:criteria-wrd-race},
  we find that $\neg \lteVC{\writeVC{x}}{\threadVC{i}}$.
  Hence, the \SHB\ algorithm reports a race when processing $f$.
\end{proof}  

\subsection{Lemma~\ref{le:time-stamp-vc-concurrent}}

\begin{lemma}[Concurrent Epoch-VC Criteria]
\label{le:time-stamp-vc-concurrent}  
Let $T$ be a well-formed trace.
Let $e,f \in \rwTx$ for some variable $x$.
Let $\thread{j}{k}$ be the epoch of $e$
and $\vcN$ the vector clock of $f$ as calculated
by running Algorithm~\ref{alg:vc-edges-conc-read-write} on $T$.
If (1) $\pos{e} < \pos{f}$ and (2) $k > \accVC{\vcN}{j}$,
then $e$ and $f$ are concurrent to each other.
\end{lemma}
\begin{proof}
  We assume the contrary.
  Suppose $\hb{f}{e}$. This immediately leads to
  a contradiction as we assume (1) $\pos{e} < \pos{f}$.
  See Lemma~\ref{le:happens-before-vs-trace-position}.

  Suppose $\hb{e}{f}$.
  Then, $e$'s thread must have been synced with $f$'s thread.
  Hence, the time stamp at position $j$ for $f$'s thread
  must be greater or equal than $k$. This is in contradiction
  to the assumption (2) $k > \accVC{\vcN}{j}$.
\end{proof}

\subsection{Proof of Proposition~\ref{prop:time-stamp-vc-concurrent}}

\begin{proof}
  From Lemma~\ref{le:time-stamp-vc-concurrent} we follow 
  that events in $\rwVC{x}$ are concurrent to each other.
  Based on the same Lemma we can argue that
  the set $\concEvt{x}$ accumulates pairs of concurrent events.
\end{proof}

\subsection{Proof of Lemma~\ref{le:no-other-conc-in-between}}

\begin{proof}
  By induction on $T$. Consider the point where $e$ is added to $\rwVC{x}$.
  We assume that $e$'s epoch is of the form $\thread{j}{k}$.
  We show that $e$ is still in $\rwVC{x}$
  at the point in time we process $f$.

  Assume the contrary. So, $e$ has been removed from $\rwVC{x}$.
  This implies that there is some $g$ such that
  $\hb{e}{g}$ and $\pos{f} > \pos{g} > \pos{e}$.
  We show that $g$ must be concurrent to $f$.

  Assume the contrary.  Suppose $\hb{g}{f}$. But then $\hb{e}{f}$
  which contradicts the assumption that $e$ and $f$ are concurrent
  to each other.
  Suppose $\hb{f}{g}$. This contradicts the fact that $\pos{f} > \pos{g}$.

  We conclude that $g$ must be concurrent to $f$.
  This is a contradiction to (3).
  Hence, $e$ has not been removed from $\rwVC{x}$.

  By assumption $e$ and $f$ are concurrent to each other.
  Then, we can argue that $k > \accVC{\threadVC{i}}{j}$
  where by assumption $\threadVC{i}$ is $f$'s vector clock
  and $e$ has the epoch $\thread{j}{k}$.
  Hence, $(e,f)$ is added to $\concEvt{x}$.
\end{proof}

\subsection{Lemma~\ref{le:time-stamp-vc-hb}}

\begin{lemma}[Happens-Before Epoch-VC Criteria]
\label{le:time-stamp-vc-hb}  
Let $T$ be a well-formed trace.
Let $e,f \in \rwTx$ for some variable $x$.
Let $\thread{j}{k}$ be the epoch of $e$
and $\vcN$ be the vector clock of $f$ as calculated
by running Algorithm~\ref{alg:vc-edges-conc-read-write} on $T$.
If (1) $\pos{e} < \pos{f}$ and (2) $k < \accVC{\vcN}{j}$,
then $\hb{e}{f}$.
\end{lemma}
\begin{proof}
  We assume the contrary.
  Suppose $\hb{f}{e}$. This immediately leads to
  a contradiction as we assume (1) $\pos{e} < \pos{f}$.
  See Lemma~\ref{le:happens-before-vs-trace-position}.

  Suppose $f$ and $e$ are concurrent to each other.
  This is also impossible because due to assumption (2) $k < \accVC{\vcN}{j}$,
  the thread event $e$ is in must have been synced with $f$'s thread.
\end{proof}

\subsection{Proof of Proposition~\ref{le:edge-sound}}

\begin{proof}
  Result follows from Lemma~\ref{le:time-stamp-vc-hb}.
  For each primitive edge $e \gtEdge f \in \edges{x}$
  we that $\pos{e} < \pos{f}$ and $k < \accVC{\vcN}{j}$
  where $\thread{j}{k}$ is $e$'s epoch and
  $\vcN$ is $f$'s vector clock.
\end{proof}

\subsection{Proof of Lemma~\ref{le:reach-all-pairs}}

\begin{proof}
  We consider the point in time event $e$ is added to $\rwVC{x}$
  when running Algorithm~\ref{alg:vc-edges-conc-read-write}.
  By the time we reach $f$, event $e$ has been removed from $\rwVC{x}$.
  Otherwise, $(e, f) \in \concEvt{x}$ which
  contradicts the assumption.

  Hence, there must be some $g_1$ in $\rwVC{x}$
  where $\pos{e} < \pos{g_1} < \pos{f}$.
  As $g_1$ has removed $e$, there must exist $e \gtEdge g_1 \in \edges{x}$.
  
  By the time we reach $f$, either $g_1$ is still in $\rwVC{x}$,
  or $g_1$ has been removed by some $g_2$
  where $g_1 \gtEdge g_2 \in \edges{x}$ and $g_2 \in \rwVC{x}$.
  As between $e$ and $f$ there can only be a finite number of events,
  we must reach some $g_n \in \rwVC{x}$
  where $g_1 \gtEdge \dots \gtEdge g_n$.
  Event $g_n$ must be concurrent to $f$.
  Otherwise, by soundness of edge constraints we conclude $\hb{e}{f}$
  which contradicts the assumption.

  Hence, $g_n$ is concurrent to $f$.
  Hence, $(g_n, f) \in \concEvt{x}$.
  Furthermore, we have that $e \gtEdge g_1 \gtEdge \dots \gtEdge g_n \in \edges{x}$.
\end{proof}

\subsection{Proof of Theorem~\ref{th:all-concs-post}}

\begin{proof}
  We first show that the construction of $\accConcs{x}$ terminates by
  showing that no pair is added twice.
  Consider $(e,f) \in \concEvt{x}$ where $g \gtEdge e$.
  We remove $(e,f)$ and add $(g,f)$.

  Do we ever encounter $(f,e)$? This is impossible as the position of first component
  is always smaller than the position of the second component.

  Do we re-encounter $(e,f)$? This implies that there must exist $g$
  such that $e \gtEdge g$ where $(g,f) \in \concEvt{x}$ (as computed
  by Algorithm~\ref{alg:vc-edges-conc-read-write}).
  By Lemma~\ref{le:no-other-conc-in-between} this is in contradiction
  to the assumption that $(e,f)$ appeared in $\concEvt{x}$.
  We conclude that the construction of $\accConcs{x}$ terminates.
  
  Pairs are kept in a total order imposed by the position of the first component.
  As shown above we never revisit pairs.
  For each $e$ any predecessor $g$ where $g \gtEdge e \in \edges{x}$
  can be found in constant time (by using a graph-based data structure).
  Then, a new pair is built in constant time.
  
  There are $O(n*n)$ pairs overall to consider.
  We conclude that the construction of $\accConcs{x}$ takes time $O(n*n)$.
  By Lemma~\ref{le:reach-all-pairs} we can guarantee that
  all pairs in $\allConcs{x}$ will be reached.
  Then, $\allConcs{x} \subseteq \accConcs{x}$.
\end{proof}

\subsection{Proof of Proposition~\ref{prop:elim-non-conc-pairs}}

\begin{proof}
  By construction $\pos{\thread{i}{k}} < \pos{\thread{j}{l}}$.
  Let $\e = \thread{i}{k}$ and $\f= \thread{j}{l}$.
  By Lemma~\ref{le:time-stamp-vc-hb} we have that $\e$ happens before $\f$.
  So, the pair $(\e, \f)$ is removed.
  
  If $\neg (k < \accVC{V}{j})$ we can argue that the time stamps cannot be equal.
  Hence, Lemma~\ref{le:time-stamp-vc-concurrent} applies
  and all remaining pairs in $\accConcs{x}$ must be concurrent to each other.
  Hence, $\accConcs{x} = \allConcs{x}$.
\end{proof}  

\section{Fork and Join}
\label{sec:fork-join}

We add fork and join to our language.

\begin{definition}[Fork and Join Events]
\bda{lcll}
e & ::= &  \ldots \mid \forkEE{i}{k} \mid \joinEE{i}{k} 
\eda
\end{definition}

We write $\forkEE{i}{k}$ to denote a fork event at position $k$
where the thread the event is in forks a new thread with thread id~$i$.
We write $\joinEE{i}{k}$ to denote a fork event at position $k$
where the thread the event is in waits for all events in the thread
with the tread id~$i$ to complete.

Similar to Definition~\ref{def:proper-acq-rel-order},
we require that fork and join events are properly ordered.
All events to be forked occur before the fork event
and the join event occurs after all events.

\begin{definition}[Proper Fork/Join Order]
  We say a trace $T$ enjoys a \emph{proper fork/join order}
  iff the following conditions are satisfied:

  \begin{itemize}
  \item For each $\thread{j}{\forkEE{i}{k}} \in T$
    we have that $\neg \exists e \in T$ such that
     $\compTID{e} = i$ and $\pos{e} < k$.
  \item For each $\thread{j}{\joinEE{i}{k}} \in T$
     we have that $\neg \exists e \in T$ such that
     $\compTID{e} = i$ and $\pos{e} > k$.
\end{itemize}    
  
\end{definition}  

We say a trace $T$ is \emph{well-formed} iff
trace positions in $T$ are accurate and $T$
enjoys a proper acquire/release order as well as a proper fork/join order.

We extend Definition~\ref{def:happens-before} as follows.

\begin{definition}[Fork/Join Happens-Before]
Let $T$ be a well-formed trace.
    \begin{description}
    \item[Fork order (FO):]
      Let $\thread{j}{\forkEE{i}{k}} \in T$.      
      Let $e \in T$ where $\compTID{f}=i$.
       Then, $\hbP{T}{\forkEE{i}{k}}{e}$.
  \item[Join order (JO):]
      Let $\thread{j}{\joinEE{i}{k}} \in T$.      
      Let $e \in T$ where $\compTID{e} = i$.
      Then, $\hbP{T}{e}{\joinEE{i}{k}}$.
  \end{description}
\end{definition}

\begin{algorithm}
  \caption{Fork and Join}\label{alg:fork-join}

\begin{algorithmic}[1]
\Procedure{join}{$i,j$}
\State $\threadVC{i} = \supVC{\threadVC{i}}{\threadVC{j}}$
\EndProcedure
\end{algorithmic}
  
\begin{algorithmic}[1]
\Procedure{fork}{$i,j$}
\State $\threadVC{j} = \threadVC{i}$
\State $\updateVC{\threadVC{j}}{j}{1}$
\State $\incC{\threadVC{i}}{i}$
\EndProcedure
\end{algorithmic}
 
\end{algorithm}

The necessary adjustments to construct the happens-before relation
based on vector clocks are shown in Algorithm~\ref{alg:fork-join}.
In case of a join, we synchronize the current thread's vector clock
with the vector clock of the to be joined events.
In case of a fork, we initialize the time stamp of the to be forked thread.

All results stated carry over as the treatment of fork/join is very similar to the treatment of acquire/release.

\section{Incomplete Edge Constraints and Elimination Step}
\label{sec:edges-overapproximate}

\begin{figure*}[tp]
\bda{lll|ll|ll|llllll}
  & \thread{1}{} & [1,0,0] & \thread{2}{} & [0,1,0] & \thread{3}{} & [0,0,1] & \lockVC{x} & \rwVC{x} & \concEvt{x} & \edges{x}
\\ \hline
1. & \writeE{x} & [1,0,0]  &&  && && \{ \thread{1}{1} \} &&
\\
2. & \lockE{y} &  [2,0,0]         &&  && &&&
\\
3. & \unlockE{y} & [2,0,0] &&  && & [2,0,0] && 
\\
4. &                       && \lockE{y} & [2,1,0] && &&&
\\
5. &                       && \unlockE{y} & [2,1,0]      && & [2,1,0] &&
\\
6. &                       && \writeE{x} & [2,2,0] && && \{ \thread{2}{2} \} & & \{ \thread{1}{1} \gtEdge \thread{2}{2} \}
\\
7. &                       &&            && \lockE{y} & [2,1,1]
\\
8. &                       &&            && \unlockE{y} & [2,1,1] & [2,1,1]
\\ 
9. &                       &&            && \writeE{x} & [2,1,2]  & & \{ \thread{2}{2}, \thread{3}{2} \} & \{ (\thread{2}{2}, \thread{3}{2}) \}
\eda
  \caption{Edge Constraints are Incomplete}
    \label{fig:epochs-edges}
\end{figure*}

Edge constraints are sound but not complete w.r.t.~the happens-before relation.
Consider two read/write events $e$ and $f$
on the same shared variable.
If event $f$ happens-before event $e$, there might not be an edge relation
$f \gtEdge^* e$.

Consider the example in Figure~\ref{fig:epochs-edges}.
The event $\writeEE{x}{1}$ belonging to $\thread{1}{1}$
happens before the event $\writeEE{x}{9}$ belonging to $\thread{3}{2}$
but there is no corresponding edge constraint.

Another observation is that the candidate pairs
obtained via the post-processing step
described by Definition~\ref{def:all-concs-post}
may not necessarily represent concurrent pairs.

For example, we find that
$\thread{1}{1} \gtEdge^* \thread{2}{2}$
and $\thread{1}{1} \gtEdge^* \thread{3}{2}$.
However, the pair $(\thread{1}{1}, \thread{3}{2})$ obtained via post-processing
does not form a conflicting (concurrent) pair of events.
The event $\writeEE{x}{1}$ belonging to $\thread{1}{1}$
happens before the event $\writeEE{x}{9}$ belonging to $\thread{3}{2}$.

To eliminate pairs such  as $(\thread{1}{1}, \thread{3}{2})$,
we use the following reasoning
as described by Proposition~\ref{prop:elim-non-conc-pairs} .
From $\writeEE{x}{9}$'s vector clock $[2,1,2]$ we extract the time stamp of thread $1$.
We find that this time stamp is greater than the time stamp
of epoch $\thread{1}{1}$.
Hence, we conclude that $\writeEE{x}{1}$ happens before
$\writeEE{x}{9}$. So, there is no race.

\section{Tracing}
For benchmarking we use tests from the Java Grande Forum (\cite{smith2001parallel}) and the DaCapo (version 9.12, \cite{Blackburn:2006:DBJ:1167473.1167488}) benchmark suite.
Many tests produce more than 100 million events in the given test case.
For data race prediction, only events on variables that are shared between threads are interesting.

For example, for the tomcat benchmark we encounter 26 million events from which only 11 million involve shared variables.
Another example is the xalan benchmark with over 62 million events and only 7 million on shared variables.
To reduce the size of the trace and make benchmarking feasible, we ignore events on unshared variables.
Without this filter, the memory consumption which would be far above 64 GB.

To detect unshared variables during the recording of the program trace, we perform the following tracing method.
The last thread and its access event are stored for each variable. If the same thread accesses the variable again,
the tracer only stores the current event. As soon as a second thread accesses the variable (last thread $\neq$ current thread),
the stored last event and the current event are written to the trace.
After encountering the first access by another thread all accesses to the variable are written to the trace, independent of the accessing thread.

This filtering method can introduce false positives due to wrongly ordered write-read dependencies and missed data races because of the ignored events.
The modified RoadRunner implementation can be found at \url{https://github.com/KaiSta/roadrunnerforshbee}.
Similar filters are used in~\cite{flanagan2010fasttrack}, \cite{Roemer:2018:HUS:3296979.3192385} and \cite{Mathur:2018:HFR:3288538:3276515},
where consecutive events by the same thread on a variable are ignored.
Like in our case, false positives due to `incomplete' traces may arise.

\section{\SHB\ Adaptation to predict WRD races}

We observe that write-read races due to write-read dependencies
can be directly obtained via an adaptation of the \SHB\ algorithm.

\begin{algorithm}
  \caption{Predicting WRD Races}\label{alg:wrd-race-checker}

\begin{algorithmic}[1]
\Procedure{write}{$i,x, k$}
\State $\raceVC{\writeVC{x}}{\threadVC{i}}$
\State $\raceVC{\readVC{x}}{\threadVC{i}}$  
\State $\lastWriteVC{x} = \threadVC{i}$
\State $\lastWriteID{x} = i$
\State $\updateVC{\writeVC{x}}{i}{\accVC{\threadVC{i}}{i}}$
\State $\incC{\threadVC{i}}{i}$
\EndProcedure
\end{algorithmic}

\begin{algorithmic}[1]
\Procedure{read}{$i,x, k$}
\State $\raceVC{\writeVC{x}}{\threadVC{i}}$
\State $\raceWRD{\accVC{\lastWriteVC{x}}{\lastWriteID{x}}}{\accVC{\threadVC{i}}{\lastWriteID{x}}}$
\State $\threadVC{i} = \supVC{\threadVC{i}}{\lastWriteVC{x}}$
\State $\updateVC{\readVC{x}}{i}{\accVC{\threadVC{i}}{i}}$
\State $\incC{\threadVC{i}}{i}$
\EndProcedure
\end{algorithmic}

\end{algorithm}

To distinguish between the different kinds of write-read races,
we adapt the \SHB\ algorithm as follows.
We additionally keep track of the thread id of the last write
via $\lastWriteID{x}$.
We write $\raceWRD{i}{j}$ as a short-hand for
``if $\neg (i \leq j)$ then write-read dependency race detected''.
The updates only affect the processing of read and write events.
See Algorithm~\ref{alg:wrd-race-checker}.
Thus, we can detect all write-read races due to write-read dependencies.

\begin{proposition}[SHB Completeness for Write-Read Dependency Races]
\label{prop:shb-completeness-wrd-race}  
  Let $T$ be a well-formed trace.
  Let $(e,f) \in \allRaces$ where $(e,f)$ is a write-read data race
  satisfying the criteria (2) in Definition~\ref{def:data-race}.
  Then, the \SHB\ algorithm reports that the read event of the pair $(e,f)$
  is in a (write) race.
\end{proposition}

\section{\SHBFORALL\ - \SHB\ Adaptation to predict all race pairs}

Adaptation of \SHB\ algorithm to predict all data races pairs (for a trace-specific schedule).
Some post-processing is necessary where we assume that for each processed event
we have its vector clock.

\begin{algorithm}
  \caption{\SHB\ algorithm adapted}\label{alg:vc-race-checker-adapt}

\begin{algorithmic}[1]
\Procedure{acquire}{$i,x$}
\State $\threadVC{i} = \supVC{\threadVC{i}}{\lockVC{x}}$
\EndProcedure
\end{algorithmic}

\begin{algorithmic}[1]
\Procedure{write}{$i,x, k$}
\State $\vcEvt = \{ (k,\threadVC{i}) \} \cup \vcEvt$
\State $\raceVC{\writeVC{x}}{\threadVC{i}}$
\State $\raceVC{\readVC{x}}{\threadVC{i}}$  
\State $\lastWriteVC{x} = \threadVC{i}$
\State $\updateVC{\writeVC{x}}{i}{\accVC{\threadVC{i}}{i}}$
\State $\incC{\threadVC{i}}{i}$
\EndProcedure
\end{algorithmic}

\begin{algorithmic}[1]
\Procedure{release}{$i,x$}
\State $\lockVC{x} = \threadVC{i}$  
\State $\incC{\threadVC{i}}{i}$
\EndProcedure
\end{algorithmic}

\begin{algorithmic}[1]
\Procedure{read}{$i,x, k$}
\State $\raceVC{\writeVC{x}}{\threadVC{i}}$    
\State $\threadVC{i} = \supVC{\threadVC{i}}{\lastWriteVC{x}}$
\State $\vcEvt = \{ (k,\threadVC{i}) \} \cup \vcEvt$  
\State $\updateVC{\readVC{x}}{i}{\accVC{\threadVC{i}}{i}}$
\State $\incC{\threadVC{i}}{i}$
\EndProcedure
\end{algorithmic}

\end{algorithm}

The Algorithm~\ref{alg:vc-race-checker-adapt}
additionally records for each event its vector clock.
For this purpose, we use the set $\vcEvt$.
This component does not appear in the original formulation of \SHB.
However, this extra component is necessary to predict
the set $\allRaces$ of all data race pairs 
under the schedulable happen-before relation
as defined by Definition~\ref{def:data-race}.

The set $\vcEvt$ is initially empty.
We use the trace position to uniquely identify each event
and thus record its associated vector clock
as pairs in $\vcEvt$.
The set $\vcEvt$ is updated for each write and read event.
All other parts remain the same as in Algorithm~\ref{alg:vc-race-checker}.

To predict all remaining races, we require some post-processing.
For each potential conflicting pair of events, read-write and write-write,
we need to check if the two events are concurrent to each other.
The set $\vcEvt$ records for each event its vector clock.
So, we need to consider all possible combinations of potentially
conflicting pairs and compare their vector clocks.

The following result follows from Theorem 4.2 stated
in\cite{Mathur:2018:HFR:3288538:3276515}.
\begin{theorem}[Soundness of \SHB\ Algorithm\cite{Mathur:2018:HFR:3288538:3276515}]
  Let $T$ be a well-formed trace.
  For each $e \in T$ where the \SHB\ algorithm reports a race,
  there exists $f \in T$ such that $(e,f) \in \allRaces$.
\end{theorem}

Instead of building all pairs of combinations of events,
we can use the (read/write) events reported by the \SHB\ algorithm as a starting point.

\begin{definition}[\SHB\ Post-Processing]
\label{def:shbp-post}  
  Let $T$ be a well-formed trace.
  Let $\vcEvt$ be the set obtained 
  by processing $T$ via the \SHB\ algorithm.
  Let $R$ be the set of events that are in a race as reported by $\raceVC{}{}$.
  We define
  \bda{lcl}
  \accR{T} & = &
  \left \{
  \ba{lcl}
  (e,f) & \mid & f \in R \wedge e \in T \wedge
  \\ & & \mbox{one write among $e,f$} \wedge
  \\ & & \pos{e} < \pos{f}
  \wedge
  \\ & & (\pos{e}, V_1), (\pos{f}, V_2) \in \vcEvt \wedge
  \\ & & \neg \lteVC{V_1}{V_2}
  \ea
  \right \}   
  \eda
\end{definition}
For each event in a race as reported by \SHB,
we search for a potential race partner that appears before this event in the trace.
We collect all such pairs in the set $\accR{T}$.
Thus, we obtain all (concurrent) data race pairs.
This is guaranteed by Theorem 4.2. from~\cite{Mathur:2018:HFR:3288538:3276515}:
For each (concurrent) conflicting pair $(e,f) \in \allRaces$
where $\pos{e} < \pos{f}$ we have that the \SHB\ algorithm reports that $f$ is in a race.
Hence, we can derive the following result.

\begin{corollary}[\SHB\ Post-Processing Concurrent Races]
  \label{cor:shb-post-conc-races}
  Let $T$ be a well-formed trace.
  Let $\vcEvt$ be the set obtained 
  by processing $T$ via the \SHB\ algorithm.
  Let $R$ be the set of events that are in a race as reported by $\raceVC{}{}$.
  Let $(e,f) \in \allRaces$ such that $e, f$ are concurrent and $\pos{e} < \pos{f}$.
  Then, $f \in R$ and there exists $e \in T$ such that
  $(e,f) \in \accR{T}$.
\end{corollary}
By construction $\accR{T} \subseteq \allRaces$.

\begin{example}
Recall the trace from Example~\ref{ex:shb-incomplete}.
The set of events in a race reported by $\raceVC{}{}$
are $\{ \writeEE{x}{3}, \readEE{x}{4}, \readEE{x}{5} \}$.
Based on the above post-processing characterized by Corollary~\ref{cor:shb-post-conc-races},
we find the following race pairs
\bda{l}
\{ (\writeEE{x}{1}, \writeEE{x}{3}), (\readEE{x}{2}, \writeEE{x}{3}),
(\writeEE{x}{1}, \readEE{x}{4}),
\\
(\writeEE{x}{1}, \readEE{x}{5}), (\writeEE{x}{3}, \readEE{x}{5}) \}
\eda
\end{example}  

Time complexity of the \SHB\ post-processing phase to predict all (concurrent)
data races is $O(n*n*k)$.
The algorithm reports $O(n)$ conflicting events.
Based on Corollary~\ref{cor:shb-post-conc-races} we
can use these conflicting events as a starting
point and scan through through the trace for race partners.
This requires time $O(n*n)$.
For each pair the comparison among their vector clocks takes time $O(k)$.
We assume that lookup of the vector clock for each event in $\vcEvt$ takes constant time.
Hence, $O(n*n*k)$.


\end{document}